\documentclass{article}

\usepackage{arxiv}

\usepackage[utf8]{inputenc} 
\usepackage[T1]{fontenc}    

\usepackage[numbers]{natbib}
\usepackage{doi}
\usepackage{subcaption}
\RequirePackage{amsthm,amsmath,amsfonts,amssymb,bm}
\usepackage[dvipsnames]{xcolor}
\definecolor{lightblue}{HTML}{00F9DE}

\RequirePackage{graphicx,url}
\usepackage{appendix}
\usepackage{fixltx2e}
\usepackage{float}
\usepackage{tikz}
\usetikzlibrary{fit,positioning}
\usetikzlibrary{bayesnet}
\usetikzlibrary{shapes,arrows}
\tikzstyle{block} = [draw, fill=white, rectangle, 
    minimum height=2cm, minimum width=0.5cm]
\numberwithin{equation}{section}
\newtheorem{theorem}{Theorem}[section]

\newtheorem{proposition}[theorem]{PROPOSITION}

\theoremstyle{remark}
\newtheorem{definition}[theorem]{Definition}
\newtheorem*{example}{Example}

\newtheorem{assumption}[theorem]{ASSUMPTION}

\DeclareFontFamily{OT1}{pzc}{}
\DeclareFontShape{OT1}{pzc}{m}{it}{<-> s * [1.10] pzcmi7t}{}
\DeclareMathAlphabet{\mathpzc}{OT1}{pzc}{m}{it}

\title{Causal chain event graphs for remedial maintenance}

\date{} 					

\author{{\hspace{1mm}Xuewen Yu}\thanks{Xuewen Yu was funded by Engineering and Physical Sciences Research Council (\emph{EPSRC}), with grant number EP/L016710/1, and the \emph{Statistics Department of the University of Warwick}.} \\
	Department of Statistics, University of Warwick\\
	MRC Biostatistics Unit, University of Cambridge\\
	\texttt{xuewen.yu@mrc-bsu.cam.ac.uk} \\
	\And
	{\hspace{1mm}Jim Q. Smith} \thanks{Professor Jim Q. Smith is supported by the
	\emph{Alan Turing Institute} and \emph{EPSRC} with grant number \emph{EP/K039628/1}.}\\
	Department of Statistics, University of Warwick\\
	The Alan Turing Institute\\
	\texttt{j.q.smith@warwick.ac.uk} \\
}

\renewcommand{\shorttitle}{Causal chain event graphs for remedial maintenance}

\begin{document}
\maketitle

\begin{abstract}
The analysis of system reliability has often benefited from graphical tools such as fault trees and Bayesian networks. In this article, instead of conventional graphical tools, we apply a probabilistic graphical model called the chain event graph (CEG) to represent the failures and processes of deterioration of a system. The CEG is derived from an event tree and can flexibly represent the unfolding of asymmetric processes. For this application we need to define a new class of formal intervention we call remedial to model causal effects of remedial maintenance. This fixes the root causes of a failure and returns the status of the system to as good as new. We demonstrate that the semantics of the CEG are rich enough to express this novel type of intervention. Furthermore through the bespoke causal algebras the CEG provides a transparent framework with which guide and express the rationale behind predictive inferences about the effects of various different types of remedial intervention. A back-door theorem is adapted to apply to these interventions to help discover when a system is only partially observed.
\end{abstract}


\section{Introduction}

Conventional graphical tools in system reliability include fault trees (FTs), Boolean decision diagrams (BDDs). An FT is a structured top-down logic diagram starting with the critical system event and decomposing successively into events whose composition or intersection can cause the top event \citep{Bedford2001}. A BDD can provide an equivalent graphical representation to an FT, where events are ordered in the same way as an event tree. Neither of these diagrams provides explicit partial (temporal) order of events nor standard statistical modeling methodologies associated with uncertainty handling and causal reasoning can be seamlessly embedded. Early work \citep{torres1998bayesian,pasquini1999comparing,cai2018application} also criticised traditional reliability analysis, stating that the deficiency of these models lies in the limitation of modelling the uncertain dependencies between failures and complex systems and suggested instead the use of a Bayesian network (BN). Both BNs and chain event graphs (CEGs) enjoy the flexibility of embedding probabilistic knowledge, managing probability propagation, inference, and performing causal analysis. Thus, these two classes of models can be used to inform decision makers or engineers about the potential effects of new policies or actions, enabling optimisation of the maintenance strategy in an efficient and effective way – making probabilistic graphical tools more appealing for reliability or risk analysis than traditional tools. 

Despite the popularity of the BN framework for exploring causal relationships, many researchers \citep{Shafer1996,Spirtes1993,Riccomagno2005} have argued that event tree based inference provides an even more flexible and expressive graph from which to explore causal relationships. Within artificial intelligence, methods based on probability trees are now widely used for various types of causal modeling to support decision and risk analyses in many different domains, see \textit{e.g.} causal discovery, decision making, and risk analysis \citep{Bunnin2021, Genewein2020, Zhao2017,Bedford2001}. 

Within the domain of reliability where the focus of inference is on explaining and repairing failure incidents in a system, the use of trees and their derivative depictions such as CEGs \citep{Riccomagno2005,Collazo2017,Collazo2018} provide a complementary method to the use of fault trees. Here a collection of paths on these graphs are used to explain the unfolding of events that might have led to the fault. One advantage of the CEG is that unlike the BN the asymmetric unfoldings of the process can be directly represented by its topology so that context-specific causal dependencies can be read from the tree \citep{Barclay2013,SmithAnderson2008,Anderson2006}. This is extremely useful for encoding explicitly the failure processes and deteriorating processes of machines. 

 \cite{Robins1986} demonstrated and \cite{Shafer1996} has long argued that causal assumptions are often easily inferred from tree-like structures because these represent explicitly the hypothesised time orderings of events intrinsic to many causal conjectures. A causal analysis can be performed around a framework of the CEG in much the same way as for the BN. However, although the BN has been successfully applied to support the causal analysis of various problems in reliability \citep{Fenton2018}, to our knowledge the more flexible framework of the CEG has yet to be applied to this domain. We show in this paper that the CEG is a much more expressive graphical representation than the BN for putative causes \citep{Anderson2006,Cowell2014,Thwaites2013,Collazo2018} and it can embed the sorts of asymmetries met in reliability models yet to be exploited. 

Previous work \citep{Thwaites2010,Thwaites2013} has proposed a generic method to translate Pearl's \textit{do}-calculus \citep{Pearl1995,Pearl2009} onto the CEGs. The atomic intervention on the BN that forces a variable to take a single value can be simply imported into the CEG where a singular manipulation on a causal BN corresponding to forcing multiple \textbf{edges} along the equivalent causal CEG to take a conditional probability one and others zero. For such conventional types of manipulations it has been discovered when these and more nuanced interventions can be identified. In particular, \cite{Thwaites2008,Thwaites2013} and \cite{Thwaites2010} formulated the back-door theorem and the front-door theorem on the CEG, analogous to what Pearl \citep{Galles1995,Pearl1995,Pearl2009} designed on the BNs. These support predictive models of how certain natural events might trigger failures until such time that they are remedied.

Causal reasoning about interventions to study the reliability of a system can inform the policy makers or the engineers about the potential effects of new policies or actions so that the maintenance strategy can be optimised in an efficient and effective way. However for models of system reliability, we may encounter complicated forms of causal mechanisms and unfamiliar types of intervention not usually studied in standard causal analyses when failure events automatically trigger remedial maintenance. For example when such remedial acts involve the replacement of failed components, they will behave as if starting from new rather than starting from a point of embedded usage as they would be in a more conventional causal intervention. The latter would be what we would need to assume were we to use conventional algebras. But perhaps even more important is that remedial interventions we study here are always designed to rectify a subset of \textbf{root causes} of a failure incident. So interventions associated with remedial maintenance are in practice nearly always very specific types of non-atomic (not singular) interventions. We demonstrate that the manipulations associated with such interventions are then likely not to be represented within the class of vanilla interventions considered in causal BNs -- they are far too symmetric. 

\cite{Yu2020} gave a brief introduction of different types of remedies. Here, we formalise these ideas and provide a detailed methodology which imports the concept of remedy and a root cause analysis into the CEG framework. In this way we are about to establish new causal algebras for the different remedial intervention regimes on CEGs. We show how to use CEGs to determine when and if so how we can measure probabilistically the effectiveness of remedies imposed by engineers to perceived faults.

The contributions of this article are three-fold. Firstly, we propose a novel approach for causal analysis which is applicable to reliability data. 
We formalise the concept of the stochastic manipulation on the CEG and develop the mathematical formulae to import it into the CEG. Moreover, we show the causal effects of stochastic manipulations can be identified on CEGs using the adapted back-door theorem. Both graphical criteria and proofs are given to support this new theorem. Thirdly, we define a new type of intervention -- the remedial intervention -- on the CEG for analysing system reliability. In particular, we emphasise useful causal concepts like ``remedy'' and ``root cause'' in reliability and translate them into algebras via CEGs to embellish the standard causal analysis.

In the next section we will show how to construct a CEG for analysing system reliability. Then using this framework to formally define a remedial intervention in Section \ref{sec3}. In Section \ref{sec4} we apply this definition to prove a number of results about whether or not the probabilistic effects of a given intervention are identifiable from information commonly available to an engineer. Section \ref{sec5} will use a simulated dataset to demonstrate how to perform a causal analysis using the techniques proposed in previous sections. 


\section{Constructing CEGs for reliability analysis}\label{sec2}
In this section, we will briefly review the elicitation process of CEGs from event trees and introduce new concepts for constructing a CEG for reliability data.

\subsection{A review of CEGs}
Consider a finite event tree $\mathcal{T}=(V_{\mathcal{T}},E_{\mathcal{T}})$ defined with vertex set $V_{\mathcal{T}}$ and edge set $E_{\mathcal{T}}$ \citep{SmithAnderson2008,Thwaites2008,Collazo2018,Smith2010,Gorgen2018}. Let $e_{v,v'}\in E_{\mathcal{T}}$ denote the directed edge emanating from the vertex $v$ and pointing to the vertex $v'$. For any vertex $v\in V_{\mathcal{T}}$, denote the set of its parents by $pa(v)=\{v'\in V_{\mathcal{T}}:e_{v',v}\in E_{\mathcal{T}}\}$, and the set of its children by $ch(v)=\{v'\in V_{\mathcal{T}}:e_{v,v'}\in E_{\mathcal{T}}\}$. The vertex without parents in $\mathcal{T}$ is called the \textbf{root} vertex of the tree, denoted by $v_0$. The set of vertices without children are the \textbf{leaves} of the tree, denoted by $L_{\mathcal{T}}\subset V_{\mathcal{T}}$. The non-leaf vertices are called the \textbf{situations} here, denoted by $S_{\mathcal{T}}=V_{\mathcal{T}}\setminus{L_{\mathcal{T}}}$. A \textbf{floret} $\mathcal{F}(v)=\{V_{\mathcal{F}(v)},E_{\mathcal{F}(v)}\}$ of a situation $v\in V_{\mathcal{T}}$ is a subtree of $\mathcal{T}$, whose vertex set $V_{\mathcal{F}(v)}$ consists of $v$ and $ch(v)$, and whose edge set is $E_{\mathcal{F}(v)}=\{e_{v,v'}\in E_{\mathcal{T}}:v'\in ch(v)\}$. A path starting from $v_0$ and terminating in $v\in L_{\mathcal{T}}$ is called a \textbf{root-to-leaf path} on the tree. Let $\Lambda_{\mathcal{T}}$ denote the collection of all the root-to-leaf paths of $\mathcal{T}$. 

If each edge $e_{v,v'}\in E_{\mathcal{T}}$ has an associated transition probability, denoted by $\theta_{v,v'}$, such that $\sum_{v'\in ch(v)}\theta_{v,v'}=1$ and $\theta_{v,v'}\in (0,1)$, then a \textbf{probability tree} with structure $\mathcal{T}$ and set of probabilities $\theta_{\mathcal{T}}=\{\bm{\theta}_{v}\}_{v\in S_{\mathcal{T}}}$ can be well-defined. The probability vector  $\bm{\theta}_v=(\theta_{v,v'})_{v'\in ch(v)}$ is defined for each situation. Then the probability of traversing along any path can be evaluated. Let $\pi(\cdot)$ represent the path related probability on the tree and $\theta_{v,v'}=\pi(v'|v)$.

A \textbf{staged tree} \citep{Collazo2018} is a coloured probability tree $(\mathcal{T},\bm{\theta}_{\mathcal{T}})$ where different colour represent different \textbf{stages}. Two situations $v,v'$ are in the same stage if the probability distributions over the set of edges $E(v),E(v')$ are the same. Let $U_{\mathcal{T}}=\{u_0,\cdots,u_n\}$ denote the set of stages of $\mathcal{T}$. Given a staged tree, two situations $v,v'$ in the same stage are in the same \textbf{position} if and only if $\mathcal{T}(v)$ and $\mathcal{T}(v')$ are isomorphic, \textit{i.e.} having same set of edges, vertices, and colouring. Let $W_{\mathcal{T}}=\{w_0,\cdots,w_m\}$ denote the set of positions. All the leaves in the staged tree belongs a sinking status, denoted by $w_{\infty}$. A \textbf{chain event graph} (CEG) can then be constructed by merging the vertices that belong to the same position and vertices in the sinking status and the corresponding edges that share the same colour. Let  $\mathcal{C}=(V_{\mathcal{C}},E_{\mathcal{C}})$ denote the graphical representation of the CEG. The vertex set is $V_{\mathcal{C}}=W_{\mathcal{T}}\cup w_{\infty}$. If edges $e_{v,v'}$,$e'_{w,w'}\in E_{\mathcal{T}}$ and $v,w$ are in the same position, then there exist corresponding edges $f, f'\in E_{\mathcal{C}}$. If also $v',w'$ are in the same position, then $f=f'$. The positions and edges retain their corresponding colour from the staged tree. The probabilities $\theta_f\in\bm{\theta}_{\mathcal{C}}$ of edge $f\in E_{\mathcal{C}}$ in the new graph are the same as the transition probabilities of the corresponding edges in the staged tree. Then a CEG is defined as $(\mathcal{C},\bm{\theta}_{\mathcal{C}})$. Let $w_0\in W_{\mathcal{T}}$ denote the root node of the CEG. The path starting from $w_0$ and terminating in $w_{\infty}$ is the \textbf{root-to-sink path}. Let $\Lambda_{\mathcal{C}}$ denote the collection of all the root-to-sink paths on $\mathcal{C}$.

\subsection{A CEG for reliability analysis}
An event tree $\mathcal{T}$ is next constructed to represent how system may fail. The sequence of events which would have occurred prior to maintenance are explicitly represented on each of the root-to-leaf path. Call the labels of events on edges by \textbf{d-events} and denote it by $X_{\mathcal{T}}=\bigcup_{e\in E_{\mathcal{T}}}x(e)$. Leaves represent the final status of the system, states are labelled by \texttt{fail} or \texttt{not fail}. We add restrictions to the definition of stages: $v$ and $v'$ are in the same stage if and only if $E(v)$ and $E(v')$ represent the same set of d-events and  $\theta_{v,v'}=\theta_{v^*,v^{*'}}$ if $x(e_{v,v'})=x(e_{v^*,v^{*'}})$. Unusually for this application of CEGs it is useful to define two sink nodes for $\mathcal{C}$: if $v\in L_{\mathcal{T}}$ represents a failed condition, then $v\in w_{\infty}^f$, otherwise $v\in w_{\infty}^{n}$ represents an operational but worn-out condition. We call these the \textbf{failure} and the \text{working sink nodes} respectively. Thus the vertex set of the CEG is $V_{\mathcal{C}}=W_{\mathcal{T}}\cup w_{\infty}^f\cup w_{\infty}^n$. For any path $\lambda\in\Lambda_{\mathcal{C}}$ that ends in the sink $w^f_{\infty}$, we call a \textbf{failure path} and represents a possible pathway to fail. All other paths -- those which terminate in $w^n_{\infty}$ are called \textbf{deteriorating paths}. Let $\Lambda_{\mathcal{C}}^f$ and $\Lambda_{\mathcal{C}}^n$ denote the sets of failure paths and deteriorating paths respectively so that $\Lambda_{\mathcal{C}}=\Lambda_{\mathcal{C}}^f\cup\Lambda_{\mathcal{C}}^n$ and $\Lambda_{\mathcal{C}}^f\cap\Lambda_{\mathcal{C}}^n=\emptyset$.

CEGs have the advantage over BNs to explicitly expressing within its topology the pathway to failure. More explicitly the chronological development of the failure or deteriorating processes can be captured by the root-to-sink paths. We can order the d-events by beginning with root causes, followed by the cascading events initiated by the root causes, \textit{i.e.} primary faults and secondary faults and so on, and ending with a failure or worn-out condition. Because a cause will always happen before an effect the order of these cascading events, expressed in the CEG, embodies the full causal story about what happens if a unit passes along one of its paths. In particular this enables us to infer the nature of a root cause. In this way we can use the topology of the CEG needed to examine the efffects of a given remedial act in fixing the root cause of a failure.

To demonstrate how such a CEG analysis begins, consider an example of failures associated with bushing systems depicted in Figure \ref{fig:ceg}. Bushings are components in a transformer for insulation. We constructed an event tree according to the investigation report by \cite{Al2017}. The tree starts with classification of causes that may lead to system failures. This is followed by depicting the potential causes, then the symptoms that might arise from these. The last component represented on the tree is the failure indicator. 
A staged tree next colours the vertices of $\mathcal{T}$ to embed various context-specific conditional independences that an engineer might bring to the study. For example, here we assume that when the environment outside the machine impacts the system negatively, system failure is conditionally independent of the exact exogenous environment. This assumption makes situations $v_7$ and $v_8$ in the same stage. The stages in Figure \ref{fig:staged tree} are $u_0=\{v_0\}, u_1=\{v_1\},u_2=\{v_2\},u_3=\{v_3,v_4\},u_4=\{v_5\},u_5=\{v_6\},u_6=\{v_7,v_8,v_{13},v_{14},v_{15},v_{16}\},u_7=\{v_9,v_{11}\},u_6=\{v_{10},v_{12}\}.$ Figure \ref{fig:ceg} presents the CEG derived from this staged tree. The positions are $w_0=\{v_0\}$, $w_1=\{v_1\}$, $w_2=\{v_2\}$, $w_3=\{v_3,v_4\},w_4=\{v_5\},w_5=\{v_6\}$, $w_6=\{v_9,v_{11}\},w_7=\{v_{10},v_{12}\}, w_8=\{v_7,v_8,v_{13},v_{14},v_{15},v_{16}\}$.

\begin{figure}
    \centering
    \begin{subfigure}[b]{0.48\textwidth}
        \includegraphics[width=\textwidth]{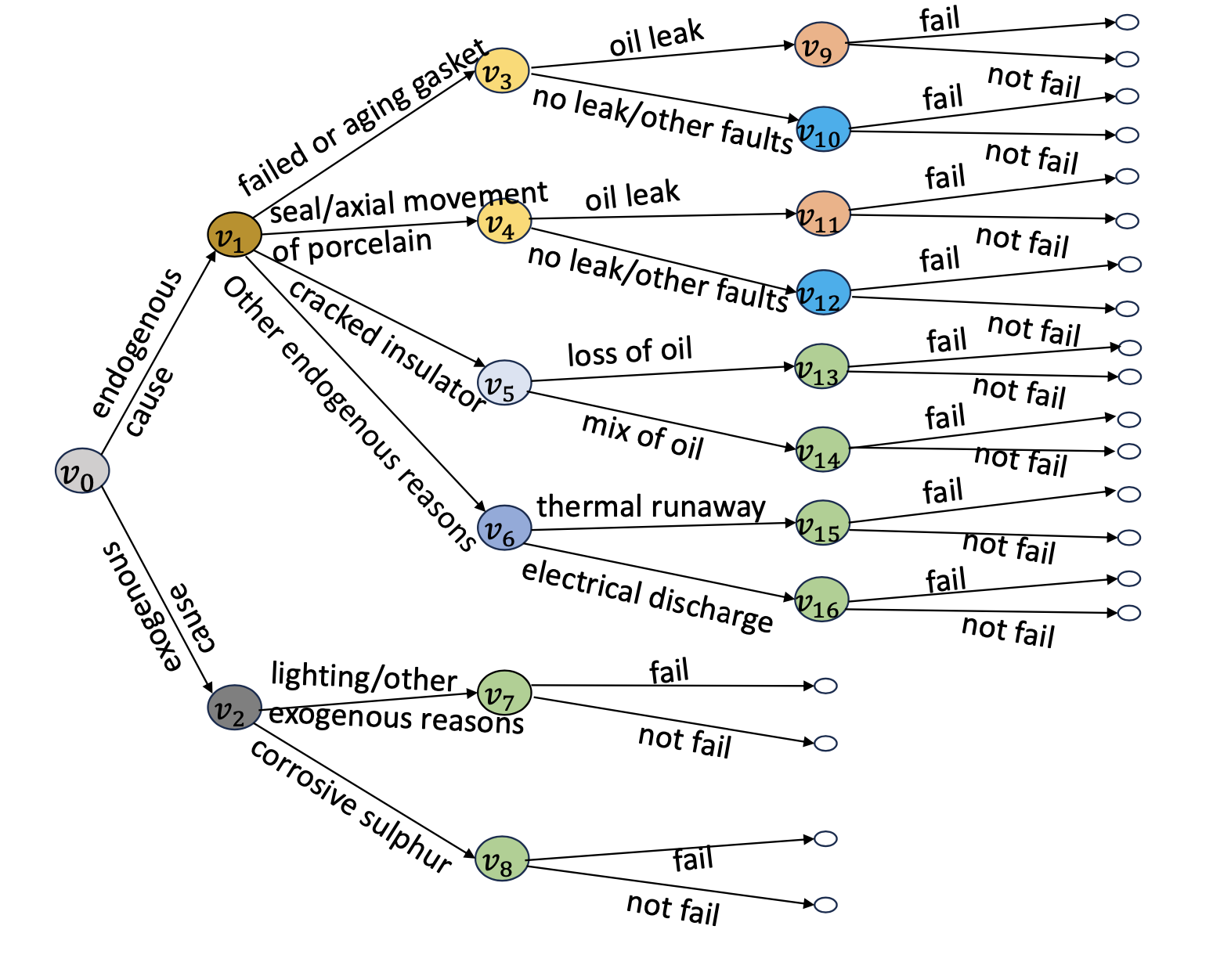}
    \caption{The hypothesised staged tree.}
    \label{fig:staged tree}
    \end{subfigure}
    \begin{subfigure}[b]{0.48\textwidth}
        \includegraphics[width=\textwidth]{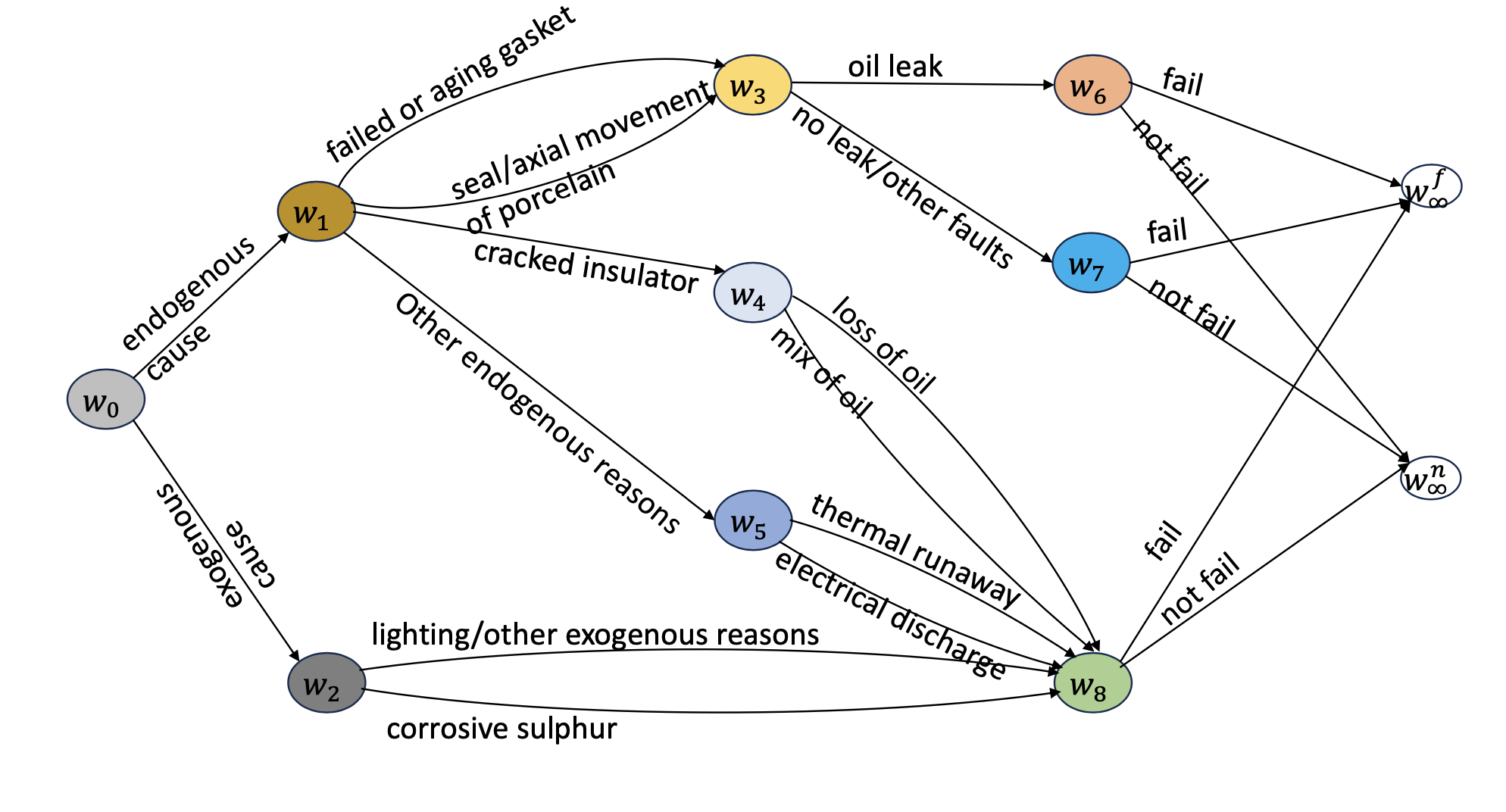}
    \caption{The CEG derived from (a).}
    \label{fig:ceg}
    \end{subfigure}
    \caption{An example of the staged tree and the CEG derived for a bushing system.}
    \label{fig:bushing trees}
\end{figure}

In practice, we can construct an event tree for a system using domain knowledge, and elicit the CEG from it based on expert judgement or making appropriate assumptions. If there are data available for this system, we can apply the structural learning algorithm \citep{Cowell2014,Collazo2018} to find the topology of the staged tree that best describes the data and derive the CEG accordingly: just as we would if modeling with a BN.  

We can now perform causal analysis on a CEG through extending Pearl's \textit{do-}operation \citep{Pearl2009}. This is defined as the \textbf{singular manipulation} \citep{Thwaites2008,Thwaites2013,Thwaites2010}.  Thwaites and others \citep{Thwaites2010,Barclay2014} formalised a \textbf{causal CEG} as follows: under a singular manipulation at a position $w$ so that the event represented on $e_{w,w^*}$ is controlled, the CEG is a manipulated CEG with $\theta_{w,w^*}=1$, $\theta_{w,w'}=0$ for $w'\in ch(w),w'\neq w^*$,
and all other $\bm{\theta}_{w''}$, $w''\in W_{\mathcal{T}},w''\neq w,$ are unchanged. Notice that because the different root-to-leaf paths of the underlying event tree are expressed explicitly within the CEG, it is possible to express explicitly within its topology a much wider range of interventions than would ever be possible just using a BN. It is this property which enables us to develop a transparent causal algebra which is particularly suited to support the study of the causes of failure in system reliability.


\section{Causal algebras for the remedial intervention}\label{sec3}

The term ``\textbf{remedial work}'' is ubiquitous to many types of engineering reports which record the maintenance of some defects or failures. Further ``\textbf{remedy}'' is a more familiar terminology in reliability engineering than ``treatment'' \citep{rubin2003}, which is commonly adopted in medical science and has a subtle different meaning. A unit must have been failed before a remedy is applied. The remedy aims to find and fix the root cause of the observed failure in order to prevent the same defect or failure reoccurring. In contrast a treatment can be applied irrespective of the state of the unit. Furthermore whilst a remedy could be a single act it is often a combination of acts taken in sequence. So the application of causal analyses are rather different in system reliability than in medicine and public health where the majority of causal analyses have traditionally been applied.  

In light of the two essential concepts, \textit{i.e.} the remedy and the root causes, we define a novel domain-specific intervention and call it the \textbf{remedial intervention}. This is a typical external intervention customised to different types of remedies. The inferential framework of the remedial intervention focuses on the discovery of root causes of a fault and the identification of a sequence of actions that will provide a remedy to that fault. Analogously to the root cause analysis, this process can be used to understand and prevent defects of a system by tracing and correcting the initial contributing factor of these defects. Here we develop bespoke causal algebras on CEGs where remedial maintenance takes centre stage. Such new algebras extend the singular manipulation which are now well established for CEGs \citep{Thwaites2010,Thwaites2013}.

\subsection{Perfect, imperfect and uncertain remedial interventions}

For a repairable system, there are three main categories of maintenance: \textbf{perfect maintenance}, \textbf{imperfect maintenance} and \textbf{minimal maintenance} \citep{Iung2005, Borgia2009}. If the status of the system after maintenance is the same as new and has the same failure rate, then the maintenance is \textbf{perfect} and the post-maintenance status is called \textbf{as good as new} (AGAN) \citep{Iung2005,Borgia2009,Bedford2001}. If the status of the system after the maintenance returns to the working order just prior to failure, then the maintenance is \textbf{minimal}, and the post-maintenance status is called \textbf{as bad as old} (ABAO) \citep{Iung2005,Borgia2009}. If the status of the system after maintenance is somewhere between ABAO and AGAN, then the maintenance is classified to be \textbf{imperfect} \citep{Iung2005}. 

To reflect these standard categories of maintenance, we accordingly give definitions to three types of remedial interventions: the \textbf{perfect remedial intervention}, the \textbf{imperfect remedial intervention}, and the \textbf{uncertain remedial intervention}. Here we use the name ``uncertain'' instead of ``minimal'' because the causal algebras we develop later in this article concern about quantifying the characteristics of the uncertainty associated with this type of intervention.

We first make the following two assumptions before formalising the concept of a remedial interventions.
\begin{assumption}\label{assumption:system}
The idle CEG\footnote{The idle CEG is the CEG that represents a system which is not intervened or manipulated.} or the event tree is faithfully constructed with respect to the domain knowledge of a particular system so that every failure process or deteriorating process that may happen in this system can be identified on the tree and every root cause and symptom are well-captured by the semantics of the tree.
\end{assumption}

\begin{assumption}\label{assumption:root cause}
  The system modelled by the CEG is repairable, and the AGAN status is reached when the root cause of the failure is completely fixed.
\end{assumption}

\begin{figure}
    \centering
    \begin{subfigure}[b]{0.32\textwidth}
        \includegraphics[width=\textwidth]{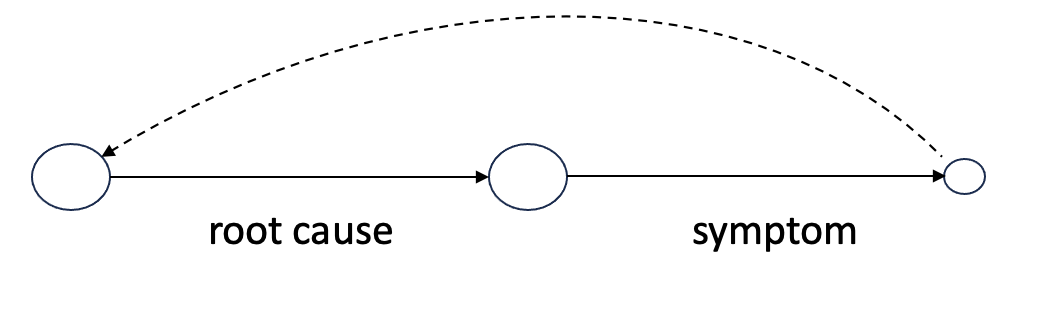}
    \caption{The perfect remedy.}
        \label{fig:remedy1}
    \end{subfigure}
    \begin{subfigure}[b]{0.32\textwidth}
        \includegraphics[width=\textwidth]{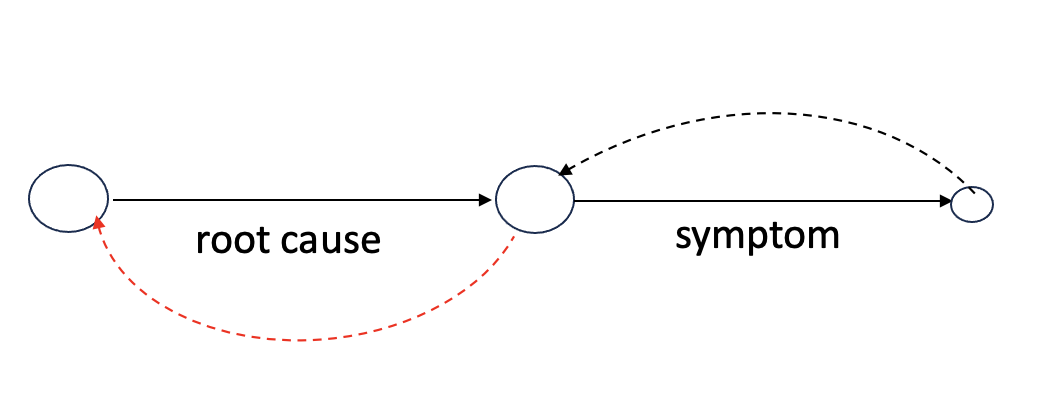}
    \caption{The imperfect remedy.}
        \label{fig:remedy2}
    \end{subfigure}
    \begin{subfigure}[b]{0.33\textwidth}
        \includegraphics[width=\textwidth]{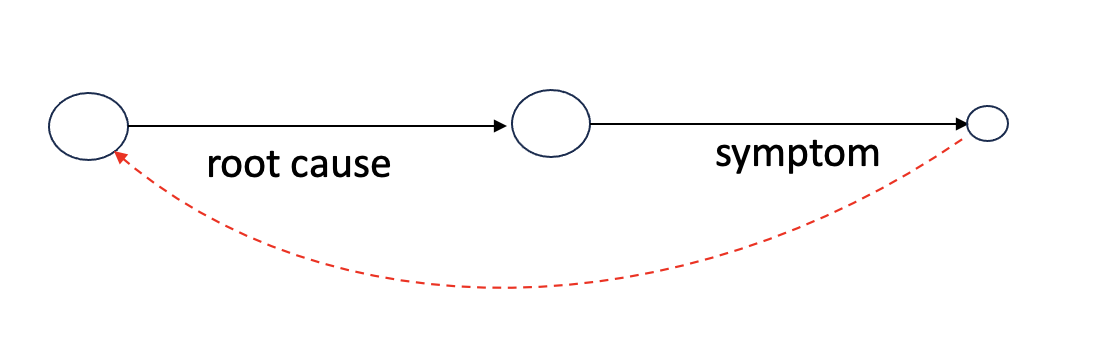}
    \caption{The uncertain remedy.}
        \label{fig:remedy3}
    \end{subfigure}
    \caption{The status monitors for the three types of remedies.}
    \label{fig:remedy}
\end{figure}

For illustrative purpose, we create a new graphical framework integrating a failure process with the process of maintenance in order to demonstrate the differences between various types of the remedial intervention, see Figure \ref{fig:remedy}. 
 
 Here we simplify the root-to-sink path by labelling only the root cause and the symptom. Take a failure process as an example. The root vertex of this path represents an AGAN status while the sink vertex of this path represents a failed condition. We call the root vertex the AGAN vertex and the leaf vertex the fail vertex. The failure path is connected by the solid black edges. The \textbf{recovery path} is defined to be the directed dashed path rooting from the fail vertex and sinking in the AGAN vertex. It models the status change of the system caused by the maintenance. The black and red dashed edges are associated with observed and unobserved maintenance respectively. Recovery paths are external to the idle CEG. This is because it represents the effect of external intervention on status of the equipment and such recovery process is not part of the description of the original system before any intervention has taken place.
 
 A remedial intervention is \textbf{perfect} if the root cause of the failure is correctly identified and successfully fixed by the observed maintenance so that the post-intervention status of the part being maintained is AGAN \citep{Yu2020,Yu2021}. The recovery process is demonstrated in Figure \ref{fig:remedy1}. The recovery path starts from the fail vertex and ends in the AGAN vertex which means the observed maintenance returns the status of the failed part to full working order. Suppose the CEG in Figure \ref{fig:ceg} failthfully models the unmanipulated bushing system, and we observe a failed bushing whose failure was caused by a cracked insulator. Then an example of a perfect remedial intervention is that the engineer replaced the cracked insulator by a new one. 

If the root cause is not remedied but only a subset of the secondary or intermediate faults are remedied, then after the intervention the status of the repaired component will not return to AGAN. However it is better than ABAO. We call such an intervention an \textbf{imperfect remedial intervention}. We can visualise the status change of the maintained equipment from Figure \ref{fig:remedy2}. The recovery path consists of a black dashed edge and a red dashed edge. The black dashed edge points from the fail vertex to the interior vertex of the failure path, which means the status of the equipment is improved but not AGAN after maintenance. In order to fully restore the system, additional maintenance is needed. If imperfect remedial work has been made at time $t$, then the maintenance log will record only that maintenance has happened. As for what is further needed to fully restore the system is unknown at that time. This brings uncertainty into this type of remedial intervention. The recovery process corresponding to the additional remedial work is represented by the red dashed edge, which points from the interior vertex to the AGAN vertex.  

If the maintenance logs do not record what remedial maintenance was taken, then such intervention is classified as an \textbf{uncertain remedial intervention}. Diagnostic information has not yet been made available so the root cause of the failure cannot be determined. A follow-up check and maintenance will be carried out in order to restore the broken part. Therefore the recovery process of this type of intervention is unobserved and so uncertain. The corresponding recovery path is shown in Figure \ref{fig:remedy3}.

\subsection{Notation and definitions} \label{sec:remedy}
To model remedial intervention, we first introduce some new variables.

Assume that the root causes of a specific defect or failure could be multiple and are well-defined. Note that a remedial intervention is defined to allow multiple root causes to be corrected simultaneously. Such an intervention is of course always non-singular within the CEG representation.

Let $\mathcal{A}$ denote the state space of maintenance events. Let $A^O$ and $A^{U}$ be random variables taking values in $\mathcal{A}$ representing observed maintenance and uncertain maintenance respectively. Let $A=(A^O,A^U)$ denote the vector of all maintenance. The status of the maintained equipment is observable, and represented by a \textbf{status indicator} $\delta$ such that \begin{equation}
    \delta=\begin{cases}
    1, & \text{ if the status is AGAN after maintenance,}\\
    0, & \text{ otherwise.}
    \end{cases}
\end{equation} 
 Let $E^{\Delta}=\{e_{l_1},\cdots,e_{l_n}\}$ denote the set of edges labelled by the d-events associated with root causes. For any edge representing a root cause $e_{l_i}\in E^{\Delta}$, we define a binary variable to indicate whether or not the root cause labelled on $e_{l_i}$ is fixed and call it an \textbf{intervention indicator}. Let 
\begin{equation}
    I_{e_{l_i}} = \begin{cases}
    1, & \text{ if the root cause represented on $I_{e_{l_i}}$ is fixed by the maintenance, }\\
    0, & \text{ otherwise.}
    \end{cases}
\end{equation} Then we have a vector of intervention indicators defined over $E^{\Delta}$, denoted by $\bm{I}_{E^{\Delta}}=\{I_{e_{l_1}},\cdots,I_{e_{l_n}}\}$. Let $\lambda^O\in\Lambda$ be the set of possible root-to-sink paths associated with the failure or deterioration. Note that the whole failure or deteriorating process might be partially observed when the root causes are unknown. Let $\lambda^R\in\Lambda$ denote the actual failure or deteriorating path when root causes are known.

When $\delta=1$, the remedial intervention is perfect. In other words, the root causes are correctly identified and fixed by $A^O$. So the value of the vector of the intervention indicators is known, denote this by $\bm{I}_{A^O}$. Then we have,
\begin{equation}
    p(\bm{I}_{E^{\Delta}}|A,\lambda^O,\delta=1) = \begin{cases}
    1, & \text{$\bm{I}_{E^{\Delta}}=\bm{I}_{A^O}$,}\\
    0, & \text{$\bm{I}_{E^{\Delta}}\neq\bm{I}_{A^O}$.}
    \end{cases} 
\end{equation} The actual paths $\lambda^R$ given $\lambda^O$, $A^O$, $\delta=1$ are deterministic: $\lambda^R=E_{A^O}\cap\lambda^O$, where $e\in E_{A^O}$ if and only if $I_{e}=1$.

When $\delta=0$, the remedial intervention is imperfect or uncertain. The uncertainty arises from the unobserved additional maintenance $A^U$. Then the root causes need to be inferred:
\begin{equation}\label{equ:imperfect or uncertain intervention}
   p(\bm{I}_{E^{\Delta}}|A^O,\lambda^O,\delta=0)=\sum_{a^u\in \mathcal{A}}p(\bm{I}_{E^{\Delta}}|a^O,a^U,\lambda^O,\delta=0)p(a^U|a^O,\lambda^O, \delta=0).
\end{equation}The probability $p(\bm{I}_{E^{\Delta}}|a^O,a^U,\lambda^O,\delta=0)\neq p(\bm{I}_{E^{\Delta}}|a^O,a^U,\lambda^O,\delta=1)$ since the latter is associated with a perfect remedy with degenerate probability distribution while the former is not. The actual path $\lambda^R$ can be inferred in the same way. Here we apply equation \eqref{equ:imperfect or uncertain intervention} for either imperfect or uncertain remedial intervention. This is because both types involve uncertainty in maintenance events, which is represented by the model $p(a^U|a^O,\lambda^O, \delta=0)$. When the intervention is uncertain, $a^O=\emptyset$ denotes an empty set and $a^U$ is informed from the partially observed failure process $\lambda^O$. In practice, we can specify a parametric model $p(a^U|a^O,\lambda^O, \delta=0; \bm{\eta}_{\lambda}^{\mathbb{I}(a^O\neq \emptyset)})$. Then $\bm{\eta}_{\lambda}^{0}$ will denote the set of parameters defined over the set of observable maintenance and the root-to-sink paths for the model under the uncertain remedial intervention, while $\bm{\eta}_{\lambda}^{1}$ will denote the set of parameters defined over the root-to-sink paths for the model under the imperfect remedial intervention.

Thus, given the observed maintenance, we can infer the probabilities associated with the intervention indicator vector as
\begin{equation}\label{equ:infer i}
\begin{split}
p(\bm{I}_{E^{\Delta}}|A^O,\lambda^O)=&p(\bm{I}_{E^{\Delta}}|\delta=1,\lambda^O,A^O)p(\delta=1|\lambda^O,A^O)+\\
&\sum_{a^u\in \mathcal{A}}p(\bm{I}_{E^{\Delta}}|a^O,a^U,\lambda^O,\delta=0)p(a^U|a^O,\lambda^O, \delta=0)p(\delta=0|\lambda^O,A^O).  
\end{split}
\end{equation}Here $p(a^U|a^O,\lambda^O, \delta=0)$ will be $0$ or near $0$ for rare events. Furthermore, under assumption \ref{assumption:system}, for any system of interest, all possible root causes are represented by the tree. Under assumption \ref{assumption:root cause}, equation \eqref{equ:infer i} enables us to identify the corresponding root causes for any remedial intervention. Therefore, equation \eqref{equ:infer i} provides a general form for modelling any remedial intervention.

\subsection{The stochastic manipulation}

Engineers address root causes to prevent the fault or failure caused by these root causes reoccurring. So it is natural to assume that the distribution over root causes are affected by the remedial intervention. We then import this idea to the idle CEG. Denote the set of positions whose emanating edges are labelled by root causes as $W^{\Delta}$. A position $w\in W^{\Delta}$ if $e_{w,w'}\in E^{\Lambda}$. The intervention indicator vector defined for position $w$ is $\bm{I}_{w}=(I_{e_{w,w'}})_{w'\in ch(w)}$. Define the \textbf{intervened position} to be the position whose floret $\mathcal{F}(w)$ is assigned a new probability distribution under an intervention. Let $\bm{w}^*$ denote the set of intervened positions. Then under a remedial intervention, $\bm{w}^*\subseteq W^{\Delta}$. If $I_{e_{w,w'}}=1$ then this means the root cause represented on $e_{w,w'}$ is intervened and $w\in\bm{w}^*$. We next formalise the manipulation of the probability distributions over $\mathcal{F}(\bm{w}^*)$.

\begin{definition}[Stochastic manipulations]\label{def:stochastic}
 A manipulation on a CEG $C$ is called \textbf{stochastic} if there exists a set of positions $\bm{w}^{*}\subseteq W$ such that,
 \begin{enumerate}
     \item  for each $w\in\bm{w}^*$, there is a well-defined map $\Gamma$ updating the transition probabilities vector $\bm{\theta}_{w}=(\theta_{w,w'})_{w'\in ch(w)}$
      \begin{equation} \label{map}
     \Gamma: \bm{\theta}_{w}\mapsto\bm{\hat{\theta}}_{w},
 \end{equation}where $\bm{\hat\theta}_{w}=(\hat{\theta}_{w,w'})_{w'\in ch(w)}$ denotes the post-intervention transition probabilities vector;
 \item the new transition probabilities vector $\bm{\hat\theta}_{w}$ satisfies $\bm{\hat{\theta}}_w\neq\bm{\theta}_{w}$, $\sum_{w'\in ch(w)}\hat{\theta}_{w,w'}=1$ and $\hat{\theta}_{w,w'}\in(0,1)$ for $w'\in ch(w)$;
 \item for a position $w\in W_{\Lambda(\bm{w}^*)}\setminus{\bm{w}^*}$, i.e. a position that lies on any of the paths passing through $\bm{w}^*$ and is not an intervened position, the corresponding transition probabilities vector remains the same as the pre-intervention version: $\bm{\hat{\theta}}_{w}=\bm{\theta}_w$, here $\Lambda(\bm{w}^*)=\cup_{w\in\bm{w}^*}\Lambda(w)$ denote the \textbf{intervened paths};
 \item for a position $w'\in ch(pa(\bm{w}^*))\setminus{\bm{w}^*}$, i.e. one that shares the same parent as $\bm{w}^*$ but which is not an intervened position, $\hat{\theta}_{pa(w'),w'}=0$.
 \end{enumerate}
\end{definition}

The simplest scenario for the manipulated transition probabilities $\bm{\hat{\theta}}_{\bm{w}^*}$ is when the values of $\bm{\hat{\theta}}_{\bm{w}^*}$ are known. But these values may not necessarily be available, in which case we have a more complex scenario where we are required to learn these values from an inferential framework.

We begin with the simplest scenario when $\bm{\hat{\theta}}_{\bm{w}^*}$ are known. In this case, the post-intervention path related probabilities can be evaluated. Let 
$\pi(\cdot)$ and $\hat\pi(\cdot)$ denote the pre- and post-intervention path related probability respectively. Let $E(\bm{w}^*)=\cup_{w\in W} E(w)$ denote the emanating edges of the intervened positions, $E_{\lambda}$ denote the set of edges traversed by the root-to-sink path $\lambda$. Then, for $\lambda\in \Lambda_{\mathcal{C}}$,

\begin{equation}\label{equ:manipulate path prob}
    \hat\pi(\lambda)=\begin{cases}
    \frac{\prod_{e\in E_{\lambda}}\theta_e}{\prod_{e'\in E(\bm{w}^*)\cap E_{\lambda}}\theta_{e'}}\times \prod_{e'\in E(\bm{w}^*)\cap E_{\lambda}}\hat\theta_{e'} & \text{if $\lambda\in\Lambda(\bm{w}^*)$,}\\
    0 & \text{otherwise.}
    \end{cases}
\end{equation}

Given the intervened paths $\Lambda(\bm{w}^*)$, let $\mathcal{\hat{C}}^{\Lambda(\bm{w}^*)}$ denote the topology of the conditioned CEG \citep{Thwaites2013} constructed with respect to $\Lambda(\bm{w}^*)$. Let $\hat{\pi}^{\Lambda(\bm{w}^*)}(\cdot)$ denote the path related probability of $\mathcal{\hat{C}}^{\Lambda(\bm{w}^*)}$, $W_{\lambda}$ denote the set of vertices traversed by $\lambda$. Then, 
\begin{equation}\label{conditional prob}
    \hat{\pi}^{\Lambda(\bm{w}^*)}(\lambda)=\prod_{\substack{w_i,w_j\in W_{\lambda},\\w_i=pa(w_j)}}\hat{\pi}^{\Lambda(\bm{w}^*)}(w_j|w_i)=\prod_{\substack{w_i,w_j\in W_{\lambda},\\w_i=pa(w_j)}}\frac{\sum_{\lambda\in\Lambda(\bm{w}^*)}\hat{\pi}(\lambda,\Lambda(e_{w_i,w_j}))}{\sum_{\lambda\in\Lambda(\bm{w}^*)}\hat{\pi}(\lambda,\Lambda(w_i))}.
\end{equation}

\begin{definition}[The manipulated CEG]\label{def:manipulated ceg} The manipulated CEG of a remedial intervention with respect to $\Lambda(\bm{w}^*)$ has a topology $\hat{\mathcal{C}}=\mathcal{\hat{C}}^{\Lambda(\bm{w}^*)}=(\hat{V},\hat{E})$ and transition probabilities $\bm{\hat{\theta}}^{*}$.
\begin{itemize}
    \item The vertex set is $\hat{V}=W_{\Lambda(\bm{w}^*)}=\cup_{\lambda\in\Lambda(\bm{w}^*)}W_{\lambda}$;
    \item the edge set is $\hat{E}=E_{\Lambda(\bm{w}^*)}=\cup_{\lambda\in\Lambda(\bm{w}^*)}E_{\lambda}$;
    \item the transition probabilities are evaluated via equation (\ref{conditional prob}), equation (\ref{equ:manipulate path prob}) and the specification in Definition \ref{def:stochastic}.
\end{itemize}
\end{definition}

We can now formalise the transformation of an idle CEG to a manipulated CEG under the remedial intervention.  Let $\zeta_{\bm{w}^*}$ denote the map indexed by the intervened positions $\bm{w}^*$ corresponding to this transformation. Then,
\begin{equation}
    \zeta_{\bm{w}^*}:(\mathcal{C},\bm{\theta})\mapsto (\hat{\mathcal{C}},\bm{\hat{\theta}}^{*}).
\end{equation} This map is well-defined if and only if the map $\Gamma$ that updates the transition probabilities, see equation (\ref{map}) in Definition \ref{def:stochastic}, is well-defined. Here the intervened positions are fixed, while as discussed in Section \ref{sec:remedy}, they might not be deterministic given the observations of remedial intervention and partial failure paths and the root causes needs to be inferred through the map $\rho:(A^O,\lambda^O)\mapsto I_{E^{\Delta}}$. This map corresponds to equation \eqref{equ:infer i}. Then $\zeta\circ\Gamma\circ\rho$ allows us to construct the manipulated CEG.

An example of a manipulated CEG is given in Figure \ref{fig:manipulated ceg}. This corresponds to the idle CEG in Figure \ref{fig:ceg} for the bushing system and an imperfect remedial intervention that did not restore the status of the machine to AGAN. Suppose defects in gasket or porcelain lead to the failure. Then the intervened position in $\bm{w}^*=\{w_1\}$ and we can explore the effect of the remedial intervention by stochastically manipulating the distrubiton over $\mathcal{F}(w_1)$. The composition of stages may change when transforming from $\mathcal{C}$ to $\mathcal{\hat{C}}$. The position $w_8$ in the idle system contains a single stage $u_7$. This stage consists of six situations that consists of situations $\{v_7,v_8,v_{13},v_{14},v_{15},v_{16}\}$. While in the manipulated CEG, vertices $v_7,v_8$ are not traversed by any path in $\Lambda(v_1)$ in the event tree, where $w_1=\{u_1\}=\{v_1\}$. The root floret $\mathcal{F}(w_0)$ is associated with the root cause classifier. Here the manipulated CEG is conditioned on $\Lambda(w_1)$, so $\hat{\pi}^{\Lambda(w_1)}(w_1|w_0)=1$ and we only concern about the endogenous causes. In fact, exogenous root causes, such as lightening, are difficult to be remedied.


The manipulated CEG is associated with an intervened model and expresses what might happen had some variables been controlled under some hypothetical intervention. It allows us to identify the effect of some form of controls, \textit{e.g.} fixing a root cause, from the observed data and interpret it causally.

\begin{figure}
    \centering
    \includegraphics[width=0.7\textwidth]{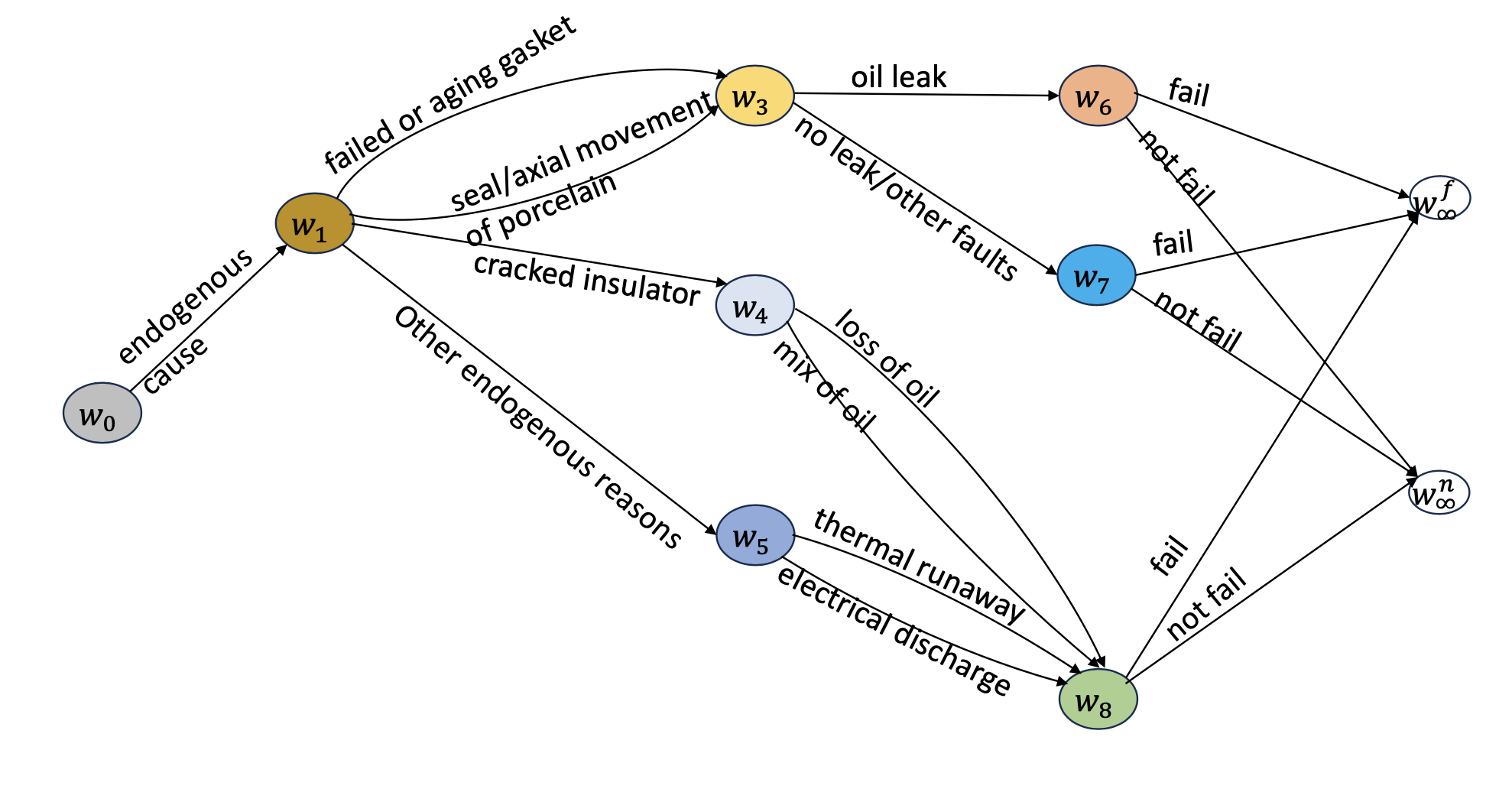}
    \caption{The manipulated CEG.}
    \label{fig:manipulated ceg}
\end{figure}

\subsection{An inferential framework using the causal algebras}

Of course in practice we need to estimate the parameters appearing in the formulae above before preceding to identify the causal effects which will be discussed in the next section. However from a Bayesian perspective this is actually straightforward. \cite{Barclay2013} and \cite{Collazo2018} have established a conjugate analysis on the non-causal CEG which translates seamlessly into this new causal setting. Here we only give an example of how Bayesian predictive inference can be performed. Let $f(\bm{\theta}_{j,w}|\bm{\alpha}_{w})$ denote the prior distribution of $\bm{\theta}_{j,w}$, which is the transition probability vector of position $w$ for individual $j$. Let $\chi_{j,e}$ denote whether d-event on edge $e$ is observed for individual $j$. The parameters of the prior is the vector $\bm{\alpha}_{w}$. Let $\bm{\gamma}_{a,\lambda}$ denote the parameter for the probability distribution of $\bm{I}_{j,E^{\Delta}}|A_j^O,\lambda_{j}^O$, with subscripts $a$ denote the maintenance and $\lambda$ denote the paths. Let $f(\bm{\gamma}_{a,\lambda}|\bm{\beta})$ denote the prior of $\bm{\gamma}_{a,\lambda}$ with prameter vector $\bm{\beta}$. Then the posterior can be written as:
 \begin{equation}\label{equ:posterior}
f(\bm{\theta}|O)\propto\prod_{j=1}^N\sum_{\bm{I}_{E^\Delta}}\prod_{e\in E/E^{\Delta}}\left(\theta_{j,e}^{\chi_{j,e}}\right)\prod_{e\in E^{\Delta}}\left(\theta_{j,e}^{I_{j,e}}p(\bm{I}_{j,E^{\Delta}}|A_j^O,\lambda_j^O,\gamma_{A_j^O,\lambda_j^O})\right)f(\gamma_{A_j^O,\lambda_j^O}|\bm{\beta})\prod_{w\in W}f(\bm{\theta}_{j,w}|\bm{\alpha}_{w}),
 \end{equation} where $O$ denotes the observations. We can sample $\bm{\gamma}_{a,\lambda}$ from the posterior and simulate the root causes through simulating $\bm{I}_{E^{\Delta}}$. Then sampling $\bm{\theta}$ from the posterior and further sampling the post-intervened probabilities by the transformation $\Gamma(\{\bm{\theta}_{w}\}_{w:\bm{I}_{w^{\Delta}}=1})$. Then we can find the predictive distribution over the paths by simulating the paths from the post-intervened transition probabilities. \cite{Yu2020} and \cite{YuEntropy2021} gave examples of implementing Bayesian inference with the customised causal algebras on CEGs for reliability analysis.


\section{Causal identifiability of a remedial intervention on the CEG}\label{sec4}
\subsection{The expression of the causal query}

For variable $X$ represented on the causal BNs, the $do-$operation \citep{Pearl2009} that forces $X$ to take value $x$ is denoted by $do(X=x)$. If we are interested in the effect of this intervention on another variable $Y$, then the causal query to be estimated is $p(y|do(x))$. This $do-$operator corresponds to the singular manipulation on the CEG that forces $\Lambda_x$ to be traversed, where $\Lambda_x$ here is the collection of root-to-sink paths passing along the edges labelled by $x$. Analogously to $p(y|do(x))$, here we identify the effect of forcing $x$ to happen on another event $y$ through estimating the probability $\pi(\Lambda_y||\Lambda_x)$. Here the notation $||$ plays a similar role as the $do-$operator which imposes an intervention onto the tree \citep{Thwaites2013}. So notationally $\pi(\Lambda_y||\Lambda_x)=\pi(\Lambda_y|do(\Lambda_x))$.

We have explained that a remedial intervention imposes a stochastic manipulation on the probability distributions $\bm{\theta}_{\bm{w}^*}$ given the intervened positions $\bm{w}^*$ and the post-intervention transition probabilities assigned to $\mathcal{F}(\bm{w}^*)$ is $\bm{\hat{\theta}}_{\bm{w}^*}$.Given $\bm{\hat{\theta}}_{\bm{w}^*}$, the causal query of a remedial intervention is
\begin{equation}\label{equ:causal1}
    \pi(\Lambda_{y}||\bm{\hat{\theta}}_{\bm{w}^*}).
    \end{equation} 
 Given an intervention $a^O$, we are interested in $\pi(\Lambda_y|do(a^O))$. However, $a^O$ is external to the system represented by the CEG. To identify this causal quantity, we transform the intervention onto the CEG using the formulation explained in the previous section. Recall that the intervened positions are identified through the map $\rho$, the transition probabilities are updated through the map $\Gamma$, and the manipulated CEG is then obtained through $\xi$. Thus, when the remedial intervention is perfect, i.e. $\delta=1$, the causal query can be expressed as
    \begin{equation}
        \pi(\Lambda_{y}||\Gamma(\left\{\bm{\theta}_{w}\right\}_{w:I_{e_{w,w'}}=1})).
    \end{equation}
When the remedial intervention is imperfect or uncertain, i.e. $\delta=0$, we estimate the causal effect
    \begin{equation}
       \sum_{\bm{I}_{E^{\Delta}}}\sum_{a^u\in \mathcal{A}}\pi(\Lambda_{y}||\Gamma(\left\{\bm{\theta}_{w}\right\}_{w:I_{e_{w,w'}}=1}))p(\bm{I}_{E^{\Delta}}|a^O,a^U,\lambda^O,\delta=0)p(a^U|a^O,\lambda^O, \delta=0), 
    \end{equation} where the post-intervention transition probabilities $\hat{\bm{\theta}}_{\bm{w}^*}$ are obtained from observations via $\Gamma\circ\rho$. For either type of remedial intervention, we need to identify $\pi(\Lambda_{y}||\Gamma(\left\{\bm{\theta}_{w}\right\}_{w:I_{e_{w,w'}}=1}))$. In this section, we only focus on the quantity $\pi(\Lambda_y||\hat{\bm{\theta}}_{\bm{w}^*}))$ with a known $\hat{\bm{\theta}}_{\bm{w}^*}=\Gamma(\left\{\bm{\theta}_{w}\right\}_{w:I_{e_{w,w'}}=1})$. Note that given the idle CEG we can construct the manipulated CEG with $\hat{\bm{\theta}}^*$. Based on this knowledge, we next show the effect of the stochastic manipulation $\pi(\Lambda_y||\hat{\bm{\theta}}_{\bm{w}^*}))$ is identifiable given the idle CEG and the observable information.

\subsection{Causal effect identifiability of stochastic manipulations}\label{sec:causal effect identifiability of stochastic manipulations}
 A \textbf{fine cut} \citep{Wilkerson2020} is defined to be a set of vertices $W'\subset W_{\mathcal{T}}$ so that $\cup_{w\in W'}\Lambda(w)=\Lambda_{\mathcal{C}}$. The intervened positions under a remedial intervention are not necessarily a fine cut. If $\bm{w}^*$ is a fine cut of $\mathcal{C}$, then the intervened paths are the set of all root-to-sink paths on the CEG, \textit{i.e.} $\Lambda(\bm{w}^*)=\Lambda_{\mathcal{C}}$. 
When the manipulations are asymmetric or the processes modelled on the idle CEG are asymmetric, $\bm{w}^*$ might not be a fine cut. A CEG conditional on the intervened paths $\Lambda(\bm{w}^*)$ can be constructed. Such a \textbf{conditioned CEG} has structure $\mathcal{C}^{\Lambda(\bm{w}^*)}=(V^*,E^*)$, where $V^*=W_{\Lambda(\bm{w}^*)}\cup w_{\infty}^{f}\cup w_{\infty}^n$ and $E^*=E_{\Lambda(\bm{w}^*)}$. The transition probabilities are $\bm{\theta}^*=\{\bm{\theta}^*_{w}\}_{w\in W_{\Lambda(\bm{w}^*)}}$, where $\theta_{w,w'}^*=\pi^{\Lambda(\bm{w}^*)}(w_j|w_i)$ which is evaluated as:\begin{equation}\label{idle conditional prob}
   \pi^{\Lambda(\bm{w}^*)}(w_j|w_i)=\frac{\sum_{\lambda\in\Lambda(\bm{w}^*)}\pi(\lambda,\Lambda(e_{w_i,w_j}))}{\sum_{\lambda\in\Lambda(\bm{w}^*)}\pi(\lambda,\Lambda(w_i))}.
\end{equation} So the conditioned CEG differs from the manipulated CEG in that it inherits pre-intervention conditional probabilities.

Analogously to the singular intervention \citep{Thwaites2010,Thwaites2013}, to identify the causal effects of a stochastic manipulation on $\bm{\theta}_{\bm{w}^*}$, we need to estimate the probability $\pi(\Lambda_y||\bm{\hat{\theta}}_{\bm{w}^*})$ from the conditioned CEG $(\mathcal{C}^{\Lambda(\bm{w}^*)},\bm{\theta}^*)$.

Let the intervened d-events be $x(E(\bm{w}^*))=\cup_{w\in\bm{w}^*}\cup_{e\in E(w)}x(e)$, which are labels of the edges emanating from the intervened positions.
For remedial interventions, they are a subset of root causes.

Given $x(E(\bm{w}^*))$ and $\bm{\hat{\theta}}_{\bm{w}^*}$, estimating $\pi(\Lambda_y||\bm{\hat{\theta}}_{\bm{w}^*})$ from $(\mathcal{C}^{\Lambda(\bm{w}^*)},\bm{\theta}^*)$ is equivalent to manipulating each d-event $x\in x(E(\bm{w}^*))$ with probability $\pi(\Lambda_x||\bm{\hat{\theta}}_{\bm{w}^*})$. Note that estimating the causal query in this way is standard for causal algebras. For example, \cite{Pearl2009} has suggested estimating the causal effects of a stochastic policy from the an unmanipulated causal BN in a similar way. Following this idea, we can formulate the causal query as follows.
\begin{definition}[Causal effect of a stochastic manipulation]\label{def:causal effects}
Given the intervened d-events $x(E(\bm{w}^*))$ and the conditioned CEG $(\mathcal{C}^{\Lambda(\bm{w}^*)},\bm{\theta}^*)$, the causal effect of the stochastic manipulation of $\bm{\hat{\theta}}_{\bm{w}^*}$ on the d-event $y$ is a function from the stochastically manipulated probability vectors $\bm{\hat{\theta}}_{\bm{w}^*}$ to the space of path related probability distributions on $\Lambda_y$ which can be expressed as:
\begin{equation}\label{equ:simp proof backdoor}
     \pi(\Lambda_y||\bm{\hat{\theta}}_{\bm{w}^*})=\sum_{x\in x(E(\bm{w}^*))}\pi^{\Lambda(\bm{w}^*)}(\Lambda_y||\Lambda_x)\hat{\pi}^{\Lambda(\bm{w}^*)}(\Lambda_x),
\end{equation} where 
\begin{equation}\label{eq:est1}
  \pi(\Lambda_y||\Lambda_x,\bm{\hat{\theta}}_{\bm{w}^*}) = \pi^{\Lambda(\bm{w}^*)}(\Lambda_y||\Lambda_x),
\end{equation}
\begin{equation}\label{eq:est2}
    \pi(\Lambda_x||\bm{\hat{\theta}}_{\bm{w}^*})=\hat{\pi}^{\Lambda(\bm{w}^*)}(\Lambda_x).
\end{equation} 
\end{definition}
Given this definition, the causal effects of such a stochastic manipulation is identifiable if and only if   $\pi^{\Lambda(\bm{w}^*)}(\Lambda_y||\Lambda_x)$ can be uniquely estimated for every $x\in x(E(\bm{w}^*))$ given the CEG and the observations. Recall that the intervened positions $\bm{w}^*$ may not form a fine cut. So $\sum_{x\in x(E(\bm{w}^*))}\pi(\Lambda_x)$ may not be equal to 1 unless conditional on $\Lambda(\bm{w}^*)$. But when the intervened positions $\bm{w}^*$ form a fine cut, then $\Lambda(\bm{w}^*)=\Lambda_{\mathcal{C}}$ and $\sum_{x\in x(E(\bm{w}^*))}\pi(\Lambda_x)=1$.

Note that $\pi^{\Lambda(\bm{w}^*)}(\Lambda_y||\Lambda_x)$ is estimated from the conditioned idle CEG $(\mathcal{C}^{\Lambda(\bm{w}^*)},\bm{\theta}^*)$. So the path related probability is $\pi(\cdot)$ instead of $\hat{\pi}(\cdot)$. On the other hand, $\hat{\pi}^{\Lambda(\bm{w}^*)}(\Lambda_x)$ is assumed to be known given $\bm{\hat{\theta}}_{\bm{w}^*}$. This is the post-intervention probability of a unit passing along the paths $\Lambda_x$ a stochastic manipulation on $\bm{\hat{\theta}}_{\bm{w}^*}$, where $\Lambda_x\subseteq\Lambda(\bm{w}^*)$. The post-intervention path related probabilities $\hat{\pi}^{\Lambda(\bm{w}^*)}(\Lambda_x)$ can be derived using equation \eqref{equ:manipulate path prob} and equation \eqref{conditional prob}.

No restriction is imposed on the intervened d-events. Under remedial interventions, however, these are normally root causes. \cite{Yu2021} also discussed another type of intervention in reliability in which case the intervened d-events can be a component, a typical symptom and so on, depending on the maintenance scheduled and conducted by the engineers. Definition \ref{def:causal effects} generically formulates a causal query from a stochastic manipulation which should not be restricted for remedial interventions.


Next we show the causal identifiability of the stochastic manipulation. 
    \begin{proposition}
    Suppose a remedial intervention imposes a stochastic manipulation on the distributions of $\mathcal{F}(\bm{w}^*)$. Then given the post-intervention transition probabilities $\bm{\hat{\theta}}_{\bm{w}^*}$, the effects of this intervention are identifiable if and only if
    $\pi^{\Lambda(\bm{w}^*)}(\Lambda_y||\Lambda_x)$ can be uniquely estimated for every $x\in x(E(\bm{w}^*))$ given the CEG and the observations.
    \end{proposition}
\begin{proof} (1) $\pi^{\Lambda(\bm{w}^*)}(\Lambda_y||\Lambda_x)\text{ is identifiable}\implies\pi(\Lambda_y||\bm{\hat{\theta}}_{\bm{w}^*})\text{ is identifiable}$:
from equation \eqref{equ:simp proof backdoor}, it is straightforward to see that if $\pi^{\Lambda(\bm{w}^*)}(\Lambda_y||\Lambda_x)$ is identifiable for all $x\in x(E(\bm{w}^*))$, then $\pi(\Lambda_y||\bm{\hat{\theta}}_{\bm{w}^*})$ can be estimated, since $\hat{\pi}^{\Lambda(\bm{w}^*)}(\Lambda_x)$ is determined externally.

(2) $\pi(\Lambda_y||\bm{\hat{\theta}}_{\bm{w}^*})\text{ is identifiable}\implies\pi^{\Lambda(\bm{w}^*)}(\Lambda_y||\Lambda_x)\text{ is identifiable}$: if there exists a d-event $x\in x(E(\bm{w}^*))$ so that $\pi^{\Lambda(\bm{w}^*)}(\Lambda_y||\Lambda_x)$ cannot be estimated, then the effects of the manipulation of $\hat{\theta}_{w,w(x)}$ are not identifiable. Since $e_{w,w(x)}\in E(\bm{w}^*)$ and $w\in \bm{w}^*$, $\hat{\theta}_{w,w(x)}\in\bm{\hat{\theta}}_{\bm{w}^*}$. Therefore the effects of $\bm{\hat{\theta}}_{\bm{w}^*}$ could not be fully estimated, so $\pi(\Lambda_y||\bm{\hat{\theta}}_{\bm{w}^*})$ could not be estimated, giving the required contradiction. So the identifiability of $\pi(\Lambda_y||\bm{\hat{\theta}}_{\bm{w}^*})$ implies the identifiability of $\pi^{\Lambda(\bm{w}^*)}(\Lambda_y||\Lambda_x)$. \\
\end{proof}

Therefore, the identifiability of $\pi^{\Lambda(\bm{w}^*)}(\Lambda_y||\Lambda_x)$ is a necessary and sufficient condition for the identifiability of $\pi(\Lambda_y||\bm{\hat{\theta}}_{\bm{w}^*})$. \cite{Thwaites2013} has proved the identifiability of a singular manipulation on the CEG through adapting the back-door theorem and the front-door theorem which are graphical tests designed by \cite{Pearl2009} to examine the identifiability of the causal effects of an atomic intervention on the causal BN. For simplicity, here we only focus on extending the back-door theorem here to prove causal identifiability. The idea is to find a partition of the root-to-sink paths $\Lambda_{\mathcal{C}}$, denoted by $\Lambda_{\bm{z}}$ so that the following criteria are satisfied.

\begin{theorem}\label{backdoor criterion} For any $e_{w,w'}\in e(x)$, if
\begin{equation}\label{equ:criterion1}
    \pi(\Lambda_z|\Lambda(w)) = \pi(\Lambda_z|\Lambda(e_{w,w'}))
\end{equation}
\begin{equation}\label{equ:criterion2}
  \pi(\Lambda_{y}|\Lambda_x,\Lambda_z)=\pi(\Lambda_{y}|\Lambda(w),\Lambda_x,\Lambda_z) = \pi(\Lambda_{y}|\Lambda(e_{w,w'}),\Lambda_z)
\end{equation} hold for every
element of $\{\Lambda_z\}$, then $\{\Lambda_z\}$ is the \textbf{back-door partition} \citep{Thwaites2013}. 
\end{theorem} 

For a stochastic manipulation, as we have explained, we wish to estimate $\pi^{\Lambda(\bm{w}^*)}(\Lambda_y||\Lambda_x)$ from $(\mathcal{C}^{\Lambda(\bm{w}^*)},\bm{\theta}^*)$. So we can simply adapt the above two criteria as follows.

\begin{theorem}\label{our backdoor} For any $w\in\bm{w}^*$, for all $x\in x(E(\bm{w}^*))$ and any $e_{w,w'}\in e(x)$ if 
\begin{equation}\label{our criterion1}
     \pi^{\Lambda(\bm{w}^*)}(\Lambda_z|\Lambda(w)) = \pi^{\Lambda(\bm{w}^*)}(\Lambda_z|\Lambda(e_{w,w'}))
\end{equation} and
\begin{equation}\label{our criterion2}
    \pi^{\Lambda(\bm{w}^*)}(\Lambda_{y}|\Lambda(w),\Lambda_x,\Lambda_z) = \pi^{\Lambda(\bm{w}^*)}(\Lambda_{y}|\Lambda(e_{w,w'}),\Lambda_z).
\end{equation}hold for every
element of $\{\Lambda_z\}$, then $\{\Lambda_z\}$ is the back-door partition for identifying the effects of a remedial intervention.
\end{theorem}
Note that $\bm{w}^*$ form a fine cut of $\Lambda(\bm{w}^*)$ and $ch(\bm{w}^*)$ also form a fine cut of $\Lambda(\bm{w}^*)$. The set of paths $\{\Lambda_{z}\}$ is a partition of $\Lambda(\bm{w}^*)$. Next we formalise the back-door theorem for identifying the causal effects of a remedial intervention.

\begin{theorem} \label{theorem}The effects of a stochastic manipulation are identifiable whenever a back-door partition $\{\Lambda_{z}\}$ can be found so that
\begin{equation}\label{use1}
    \pi(\Lambda_y||\bm{\hat{\theta}}_{\bm{w}^*})=\sum_{x\in x(E(\bm{w}^*))}\sum_{z}\pi^{\Lambda(\bm{w}^*)}(\Lambda_y|\Lambda_x,\Lambda_z)\pi^{\Lambda(\bm{w}^*)}(\Lambda_z)\hat{\pi}^{\Lambda(\bm{w}^*)}(\Lambda_x).
\end{equation}
This holds when
\begin{equation}
  \pi^{\Lambda(\bm{w}^*)}(\Lambda_y|\hat{\Lambda}_x,\Lambda_z)=  \pi^{\Lambda(\bm{w}^*)}(\Lambda_y|\Lambda_x,\Lambda_z),
\end{equation}here $\hat{\Lambda}_x$ denotes that there is a singular intervention on $\Lambda_x$ only, not on $\Lambda_z$, and
\begin{equation}
   \pi^{\Lambda(\bm{w}^*)}(\Lambda_z||\Lambda_x)= \pi^{\Lambda(\bm{w}^*)}(\Lambda_z).
\end{equation}

\end{theorem}
 It is now straightforward to deduce this from the two criteria in Theorem \ref{our backdoor} and adapting the proof of the back-door theorem for the singular intervention \citep{Thwaites2013}. We prove this theorem in the supplementary material.
  We allow flexibility in choosing $\bm{z}$, where $\bm{z}$ can be a set of d-events $\bm{z}=\{z_1,\cdots,z_l\}$. Then $\Lambda_{\bm{z}}=\{\Lambda_{z_1},\cdots,\Lambda_{z_l}\}$. We can also define $\bm{z}$ to be a set of stages, positions or edges. Next, we give an example to show how to find an appropriate back-door partition for a remedial intervention.

We can now illustrate how the formulae work for the bushing example. 
 \begin{example}Recall that we let the intervened position be $w_1$ and the manipulated CEG is Figure \ref{fig:manipulated ceg}. The intervened paths are $\Lambda_{x(e_{w_1,w_3}^1)}\cup\Lambda_{x(e_{w_1,w_3}^2)}\cup\Lambda_{x(e_{w_1,w_4})}\cup\Lambda_{x(e_{w_1,w_5})}$\footnote{The superscript indexes the order of the edge connecting $w_1$ to $w_3$ from top to bottom.}. Suppose we are interested in how the maintenance will affect system failure. Then $\Lambda_y=\Lambda_{\text{fail}}$.

Let $\Lambda_{\bm{z}}=\{\Lambda_{z_1},\Lambda_{z_2}\}$ be a partition of $\Lambda(w_1)$. We can define this partition according to the symptoms. Let $z_1$ be a set of d-events: $\{$oil leak, loss of oil, thermal runaway$\}$, \textit{i.e.} $e(z_1)=\{e_{w_3,w_6},e_{w_4,w_8}^1,e_{w_5,w_8}^1\}$. Let $z_2$ be a set of d-events: $\{$no oil leak and no other faults, mix of oil, electrical discharge$\}$. Then $e(z_2)=\{e_{w_3,w_7},e_{w_4,w_8}^2,e_{w_5,w_8}^2\}$. Thus, $ \Lambda_{z_1}=\Lambda_{x(e_{w_3,w_6})}\cup\Lambda_{x(e_{w_4,w_8}^1)}\cup\Lambda_{x(e_{w_5,w_8}^1)}$, $ \Lambda_{z_2}=\Lambda_{x(e_{w_3,w_7})}\cup\Lambda_{x(e_{w_4,w_8}^2)}\cup\Lambda_{x(e_{w_5,w_8}^2)}$.

Suppose the post-intervention probabilities $\bm{\hat{\theta}}_{w_1}$ are known. Having a stochastic manipulation on the distribution over $\mathcal{F}(w_1)$ can be treated as having a singular manipulation on $e_{w_1,w'}$ with probability $\hat{\theta}_{w_1,w'}$, where $w'\in ch(w_1)$. We can validate Theorem \ref{our backdoor} for all $x\in x(E(w_1))$. To check the first criterion, 
\begin{equation}
\begin{split}
\pi^{\Lambda(w_1)}(\Lambda(e(z_1))|\Lambda_{x(e_{w_1,w_3}^1)})
&=\frac{\theta^*_{w_0,w_1}\theta_{w_1,w_3}^{*1}\theta_{w_{3},w_6}^*}{\theta^*_{w_0,w_1}\theta_{w_1,w_3}^{*1}}= \theta_{w_{3},w_6}^*.
\end{split}
\end{equation} We can also compute:
\begin{equation}\label{eq:probs}
\begin{split}
     \pi^{\Lambda(w_1)}(\Lambda(e(z_1))|\Lambda(w_1))&=\frac{\theta^*_{w_0,w_1}(\theta_{w_1,w_3}^{*1}+\theta_{w_1,w_3}^{*2})\theta_{w_{3},w_6}^*+\theta^*_{w_0,w_1}\theta^*_{w_1,w_4}\theta_{w_{4},w_8}^{*1}+\theta^*_{w_0,w_1}\theta^*_{w_1,w_5}\theta_{w_{5},w_8}^{*1}}{\theta^*_{w_0,w_1}} \\
     &= (\theta_{w_1,w_3}^{*1}+\theta_{w_1,w_3}^{*2})\theta_{w_{3},w_6}^*+\theta^*_{w_1,w_4}\theta_{w_{4},w_8}^{*1}+\theta^*_{w_1,w_5}\theta_{w_{5},w_8}^{*1}= \theta_{w_{3},w_6}^*.
\end{split}
\end{equation}The last equality is valid because $e_{w_3,w_6},e_{w_4,w_8}^1,e_{w_5,w_8}^1$ have the same transition probabilities and $\theta_{w_1,w_3}^{1*}+\theta_{w_1,w_3}^{2*}+\theta^*_{w_1,w_4}+\theta^*_{w_1,w_5}=1$. 
So,
\begin{equation}
  \pi^{\Lambda(w_1)}(\Lambda(e(z_1))|\Lambda(e_{w_1,w_3}^1))= \pi^{\Lambda(w_1)}(\Lambda(e(z_1))|\Lambda(w_1)).
\end{equation} Following in this way, it is easy to check that the first criterion is satisfied for all manipulated $x$ and partition $z$. 

Now we check whether the second criterion is satisfied.
\begin{equation}
\begin{split}
      \pi^{\Lambda(w_1)}(\Lambda_{\text{fail}}|\Lambda(e(z_1)),\Lambda_{x(e_{w_1,w_3}^1)})&= \pi^{\Lambda(w_1)}(\Lambda_{\text{fail}}|\Lambda(e(z_1))\cap\Lambda_{x(e_{w_1,w_3}^1)}) \\
      &=\pi^{\Lambda(w_1)}(\Lambda_{\text{fail}}|\Lambda(e_{w_3,w_6}),\Lambda(e_{w_1,w_3}^1))\\
      &=\pi^{\Lambda(w_1)}(\Lambda_{\text{fail}}|\Lambda(e(z_1)),\Lambda(e_{w_1,w_3}^1)). 
\end{split}
\end{equation} Using the same method, it is straightforward to check the second criterion is satisfied for $e(z_1),e(z_2)$ and all the intervened events. Therefore, $\Lambda_{\bm{z}}$ is a back-door partition.
\end{example}

As mentioned in the previous section, there is a special case of the stochastic manipulation that $\bm{w}^*$ form a fine cut of the idle CEG. In this case, 
\begin{equation}
    \pi(\Lambda_y||\bm{\hat{\theta}}_{\bm{w}^*})=\sum_{x\in x(E(\bm{w}^*))}\sum_{z}\pi(\Lambda_y|\Lambda_x,\Lambda_z)\pi(\Lambda_z)\hat{\pi}^{\Lambda(\bm{w}^*)}(\Lambda_x).
\end{equation} Then to identify the causal effects, we only need to show that for every $x\in x(E(\bm{w}^*))$ we can find a back-door partition satisfying the criteria in Theorem \ref{backdoor criterion}.

The formulae given in this section is useful because we can formally identify the efficacy of the intervention which is expressed through a quantitative probability score which is a function of terms we can identify from the idle system.

\begin{figure}
    \centering
    \includegraphics[width=0.7\textwidth]{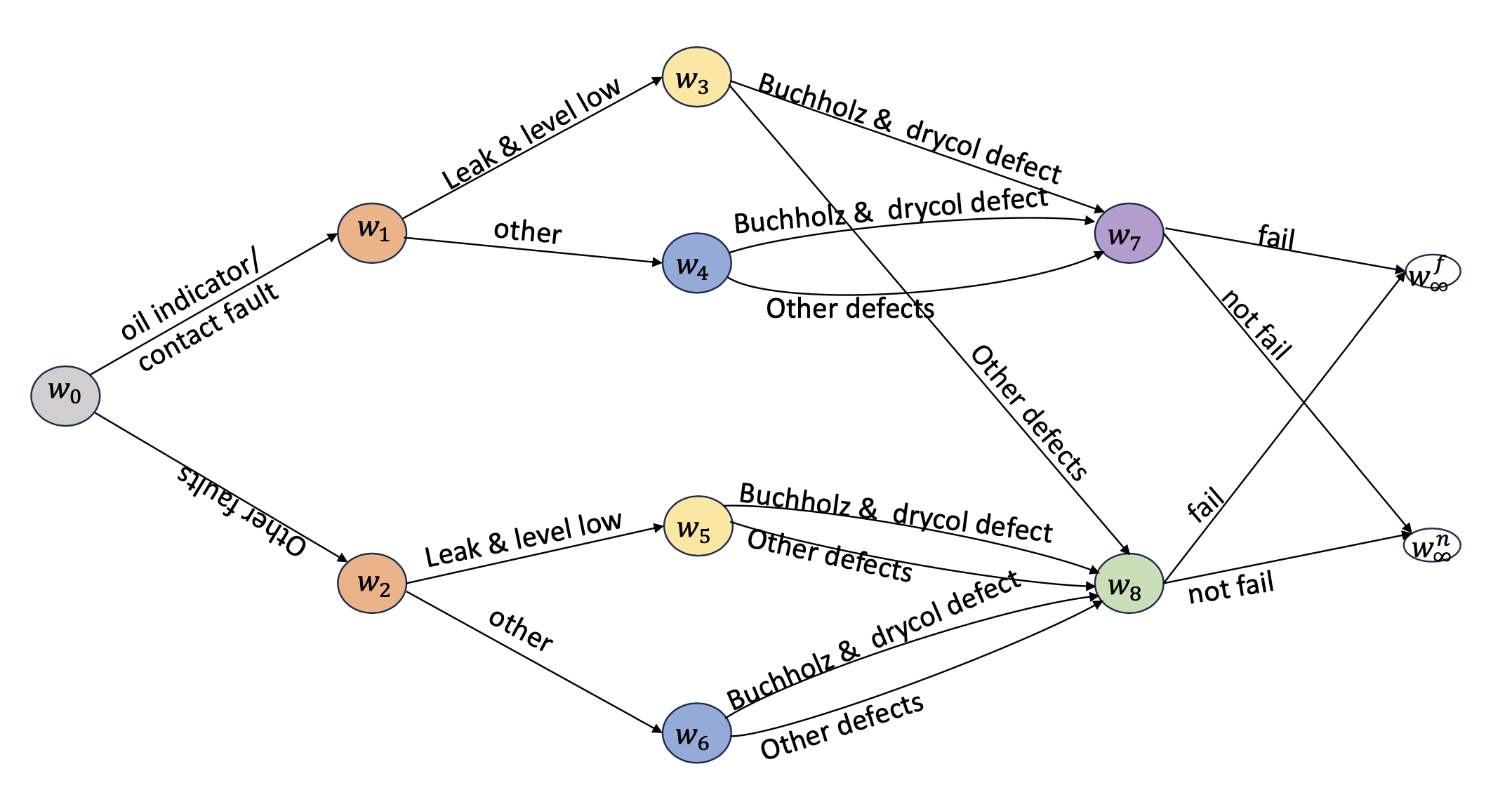}
    \caption{A hypothesised CEG of a conservator system for the example in Section \ref{sec:causal effect identifiability of stochastic manipulations}.}
    \label{fig:conservator ceg}
\end{figure}

\begin{example}\label{eg:con} Here we give an example of a conservator system of a transformer, see Figure \ref{fig:conservator ceg}. The initial events are root causes: $\{$\texttt{oil indicator/contact fault}, \texttt{other fault}$\}$. Following the root causes, we attach the oil status of the transformer: $\{$\texttt{leak} $\&$ \texttt{level low}, \texttt{other} $\}$, where \texttt{other} refers to the condition when there is no leak and oil level is normal, or there is only oil leak, or only oil level is low. Defects in two components buchholz and drycol may occur after oil problems. Here we consider whether the two components are both faulty or otherwise -- either is faulty or both are functioning. Then we attach the failure indicator. 

Suppose there is a remedial intervention which replaced the deteriorated seal. This maintenance remedied the contact fault. In response to the intervention, the conditional probability assigned to $e_{w_0,w_1}$ is expected to decrease and the probability distribution over $\mathcal{F}(w_0)$ is manipulated. 

The intervened paths are $\Lambda(w_0)=\Lambda_{\mathcal{C}}$. So the manipulated CEG has the same topology as the idle CEG. Let the effect event be system failure, \textit{i.e.} $\Lambda_y=\Lambda_{\text{fail}}$. We estimate the effect through:
\begin{equation}
\pi(\Lambda_{\text{fail}}||\bm{\hat{\theta}}_{w_0})=\pi(\Lambda_{x_{f,1}}||\Lambda(e_{w_0,w_1}))\hat{\theta}_{w_0,w_1}+\pi(\Lambda_{x_{f,1}}||\Lambda(e_{w_0,w_2}))\hat{\theta}_{w_0,w_2}.
\end{equation} 

We can define the back-door partition in terms of stages. The buchholz and drycol status is represented by the florets: $\mathcal{F}(w_3),\mathcal{F}(w_4),\mathcal{F}(w_5),\mathcal{F}(w_6)$. Note that $w_3$ and $w_5$ are in the same stage, while $w_4$ and $w_6$ are in the same stage. These stages lie upstream of $e(\text{fail})$ and downstream of the intervened root causes. Let $\Lambda_{z_1}=\{\Lambda(w_3),\Lambda(w_5)\}$ and $\Lambda_{z_2}=\{\Lambda(w_4),\Lambda(w_6)\}$. Then $\Lambda_{\bm{z}}=\{\Lambda_{z_1},\Lambda_{z_2}\}$ partitions $\Lambda_{\mathcal{C}}$ and each $z_i$ is associated with a single stage.

We check the first criterion by showing for any $z_i$, 
\begin{equation}
    \pi(\Lambda_{z_i}|w_0)=\pi(\Lambda_{z_i}).
\end{equation} For $z_1$,
\begin{equation}
\pi(\Lambda_{z_1})=\theta_{w_0,w_1}\theta_{w_1,w_3}+\theta_{w_0,w_2}\theta_{w_2,w_5}=\theta_{w_1,w_3}=\theta_{w_2,w_5}.
\end{equation} Given a singular intervention on  $e_{w_0,w_1}$,$\theta_{w_1,w_3}=\pi(\Lambda_{z_1}|\Lambda(e_{w_0,w_1}))$, and $\theta_{w_2,w_5}=\pi(\Lambda_{z_1}|\Lambda(e_{w_0,w_2}))$. Similarly,  it is easy to show that $\pi(\Lambda_{z_2}|\Lambda(e_{w_0,w_1}))=\pi(\Lambda_{z_2})$ and $\pi(\Lambda_{z_2}|\Lambda(e_{w_0,w_2}))=\pi(\Lambda_{z_2}).$ So the first criterion is satisfied.

When there is a singular intervention on  $e_{w_0,w_1}$, 
\begin{equation}
\pi(\Lambda_{\text{fail}}|\Lambda(w_0),\Lambda_{x},\Lambda_{z_1}))=\pi(\Lambda_{\text{fail}}|\Lambda(w_0),\Lambda(e_{w_0,w_1}),\Lambda_{z_1}))=\pi(\Lambda_{\text{fail}}|\Lambda(e_{w_0,w_1}),\Lambda_{z_1})),
\end{equation}and
\begin{equation}
    \pi(\Lambda_{\text{fail}}|\Lambda(w_0),\Lambda_{x},\Lambda_{z_2}))=\pi(\Lambda_{\text{fail}}|\Lambda(w_0),\Lambda(e_{w_0,w_1}),\Lambda_{z_1}))=\pi(\Lambda_{\text{fail}}|\Lambda(e_{w_0,w_1}),\Lambda_{z_2})).
\end{equation} Similarly, when there is a singular intervention on  $e_{w_0,w_2}$, we reach the same expressions as above with $e_{w_0,w_1}$ replaced by $e_{w_0,w_2}$. So the second criterion in Theorem \ref{our backdoor} is satisfied. Therefore, $\{\Lambda_{z_1},\Lambda_{z_2}\}$ forms a back-door partition here.
\end{example}

\begin{figure}
    \centering
    \begin{subfigure}[b]{0.45\textwidth}
        \includegraphics[width=\textwidth]{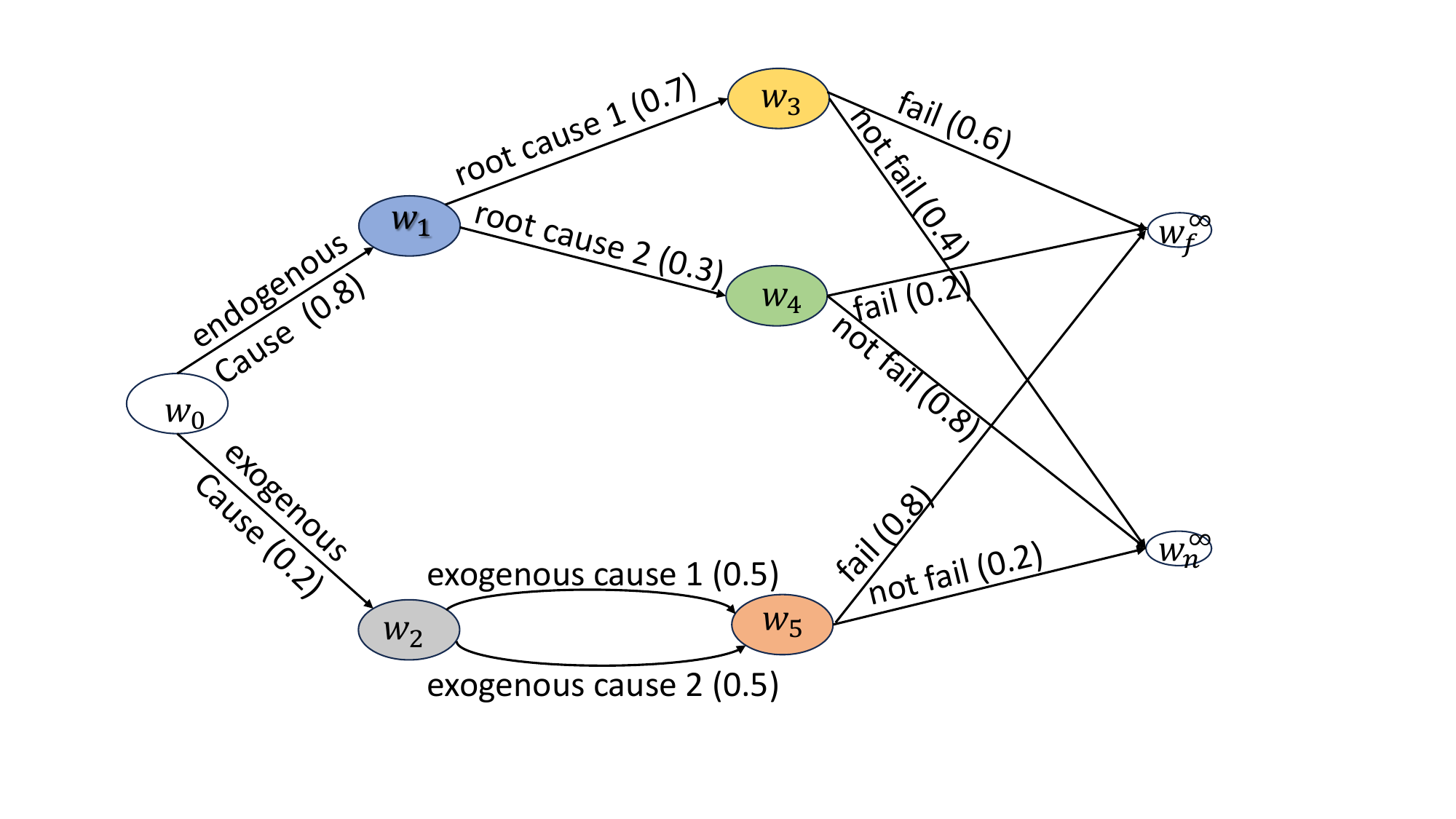}
    \caption{The true causal CEG. }
    \label{fig:true casual ceg}
    \end{subfigure}
    \hfill
    \begin{subfigure}[b]{0.45\textwidth}
    \includegraphics[width=\textwidth]{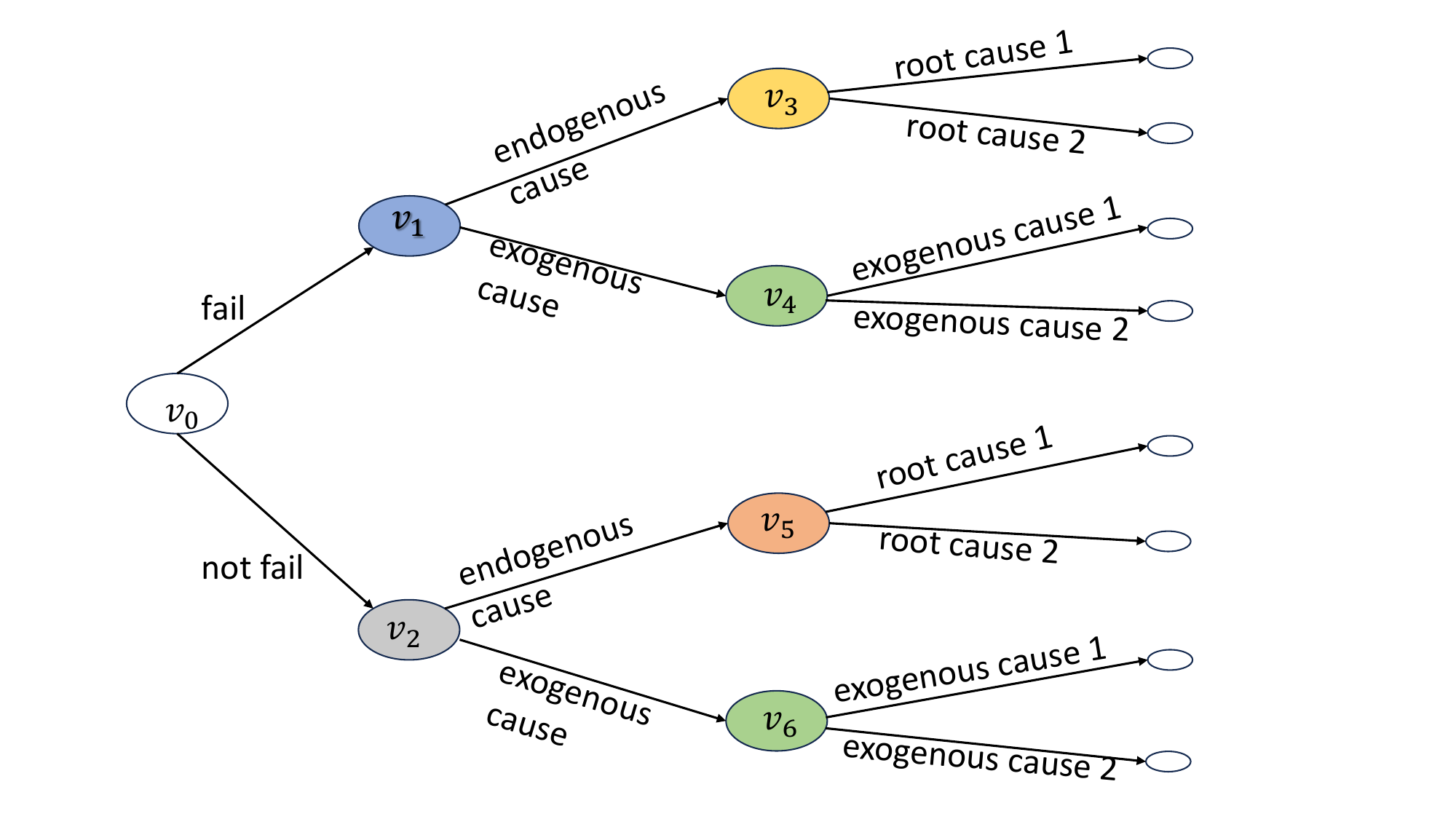}
        \caption{The learning tree.}
        \label{fig:learning tree}
    \end{subfigure}
    \caption{Trees for the simulation study. Numbers in brackets are the true parameters. We use the same colour palette for the trees. Vertices in Figure \ref{fig:true casual ceg} and Figure \ref{fig:learning tree} with the same colour do not share the same transition probability vector.}
\end{figure}

\section{A simulation study}\label{sec5}
The engineer reports in the domain we studied are sensitive. So instead we demonstrate how to apply the causal algebras proposed in previous sections in practice on a simulated dataset.\footnote{The R code implementing the analysis is available at \url{https://github.com/xwyu2021}. Other packages for CEGs are available, including R packages \texttt{stagedtrees} \citep{carli2020r}, \texttt{ceg} \citep{collazo2017ceg}, and the python package \texttt{cegpy} \citep{walley2023cegpy}.} Let Figure \ref{fig:true casual ceg} be the ground truth causal CEG \footnote{Here for simplicity we do not consider symptoms.} and we simulate a synthetic dataset comprising of $5000$ cases. There are $2698$ failure cases which emulate the information extracted from the failure report. The rest emulate the record of preventive maintenance \citep{YuEntropy2021}. We generate perfect remedies by simply assuming 
\begin{align}
  p(\delta=0|\text{endogenous cause},\text{fail},\text{root cause 1})&=  p(\delta=0|\text{endogenous cause},\text{fail},\text{root cause 2})\\
  &=p(\delta=0|\text{endogenous cause},\text{fail})=0.3.
\end{align}
 For the cases with non-perfect remedy we assume the failure processes are only partially observed.

In practice the explanations of the effects of remedial interventions will have been extracted from engineer's reports either manually or automatically \citep{Yu2022}. The tree consistent with these explanations and data should begin with the failure indicator. Thus we transform Figure \ref{fig:true casual ceg} to an equivalent tree in Figure \ref{fig:learning tree}, where the transition probabilities can be calculated using Bayes rule. We call it the \textbf{learning tree}.

Suppose we know the ground truth event tree for this system, which shares the same the topology as the learning tree. We next learn the parameter and stages of the tree from the simulated data. The parameters can be learned in a Bayesian framework based on equation \ref{equ:posterior} using the Metropolis-Hastings algorithm. Here we simply assume Dirichlet priors for transition probability vectors, where $\bm{\theta}_w\sim\text{Dirichlet}(\bm{\alpha}_w)$. The choice of $\bm{\alpha}_w$ is discussed in the supplementary material. The structure of the CEG (i.e. stages) is learned using the maximum a posterior (MAP) algorithm -- this is the most used Bayesian structural learning method for CEGs, see \citep{Barclay2013, Cowell2014, Yu2020, YuEntropy2021,strong2023methodological}. After learning the best-scored CEG on the learning tree, we transform it back to the causal CEG. This allow us to perform the causal analysis described in Section \ref{sec4}. To impose the effects of remedies, we manipulate $\bm{\theta}_{w_1}$ via $\Gamma$. Here we simply let $\frac{\hat{\theta}_{w_1,w_3}}{\hat{\theta}_{w_1,w_4}}=\frac{\theta_{w_1,w_3}}{\theta_{w_1,w_4}}\times 0.8^{\frac{N_{w_1,w_3}-N_{w_1,w_4}}{N}}$, where $N_{w_i,w_j}$ counts the observations associated with $e_{w_i,w_j}$. Other choices of $\Gamma$ have been discussed by \cite{Yu2020} and\cite{Yu2022}.

The best-scored tree learned from the synthetic data has the same structure as Figure \ref{fig:learning tree} where $v_3$ and $v_5$ are in different stages while $v_4$ and $v_8$ are in the same stage. Its MAP score is $-7212.494$,  $2698.319$ higher than the score of the tree with $v_3$ and $v_5$ in the same stage. The sensitivity analysis is summarised in Section B of the supplementary material. The manipulation shifts the distribution of $\theta_{w_1,w_3}$ to the left, as shown in Figure \ref{fig:manipulation}. The posterior mean of $\theta_{w_1,w_3}$ is reduced to $0.492$ from $0.503$ after manipulation. We can then predict the system failure induced by the two root causes based on the estimated distribution of $\hat{\bm{\theta}}_{w_1}$, see Figure \ref{fig:prediction}.

\begin{figure}
    \centering
    \begin{subfigure}[b]{0.43\textwidth}
\includegraphics[width=\textwidth]{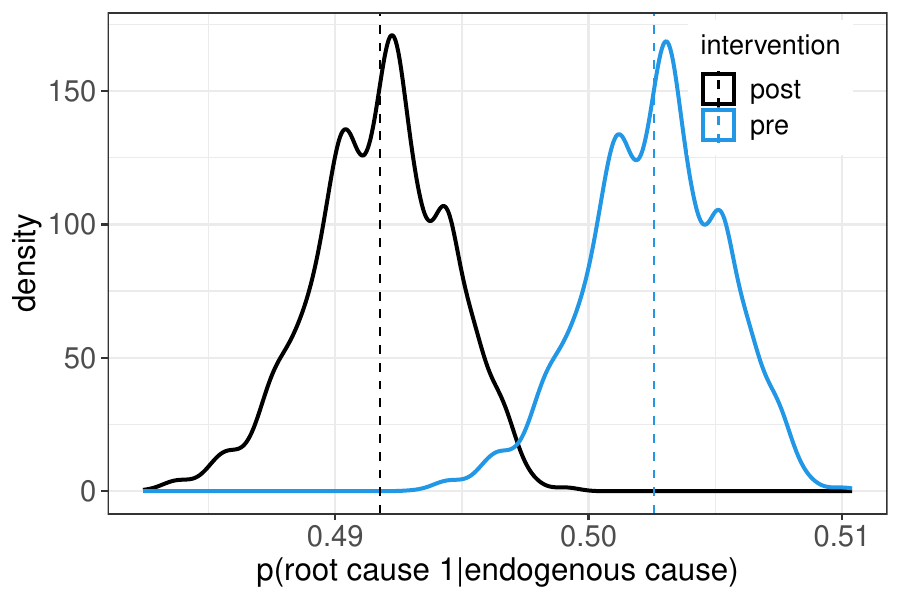}
    \caption{Distribution of $\theta_{w_1,w_3}$ before and after intervention.}
    \label{fig:manipulation}
    \end{subfigure}
    \hfill
    \begin{subfigure}[b]{0.43\textwidth}
\includegraphics[width=\textwidth]{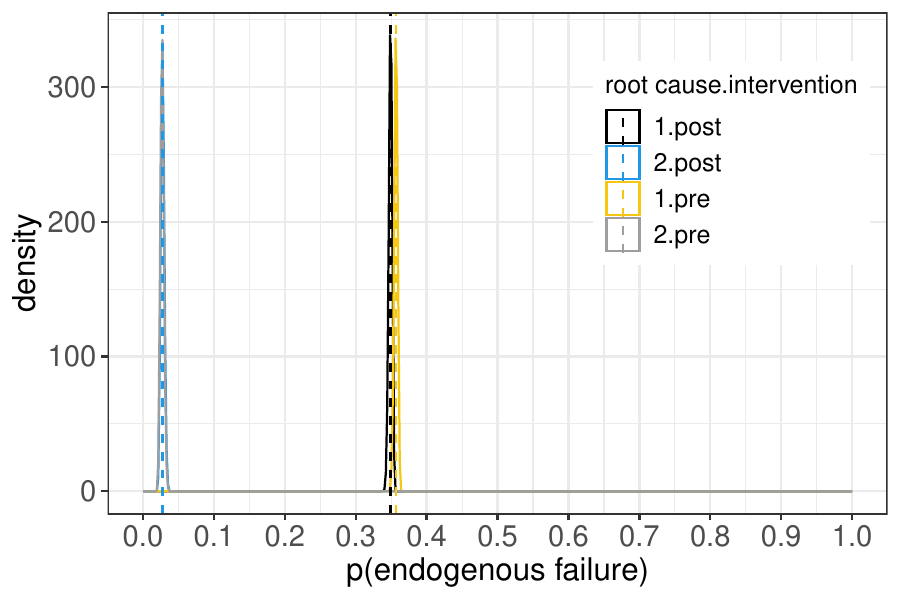}
        \caption{Predictive distribution of endogenous failure induced by different root causes.}
        \label{fig:prediction}
    \end{subfigure}
    \caption{Results of simulation study.}
    \label{fig:results}
\end{figure}


\section{Discussion}

In this article, we have shown the flexibility of the semantics of a CEG in representing asymmetric processes and capturing the effects of various interventions even when the manipulations are asymmetric. Given a context-specific CEG for a particular system, we can design the bespoke causal algebras for different types of remedial interventions. We can predict a machine's failure probability by imposing the underlying stochastic manipulations to the idle system so that we can identify the effects of the remedial intervention through finding appropriate back-door partition on the CEG. The graphical methodology we have described here therefore provides an excellent framework for translating established formal causal analyses so that these can be embedded into mainstream system reliability.

The original domain that motivated this formal development did not contain the case where the discovery of a fault would encourage an improvement of the system. However, a referee pointed out there are many domains where a fault would provoke a system upgrade \footnote{Section C in the supplementary material provides a more detailed discussion on this point.}. In such cases, the failure rate could be reduced over the entire lifetime after maintenance. Thus more complex manipulations, beyond restoring root causes to AGAN, should be considered. Semi-Markov processes can be modelled on dynamic CEGs \citep{Barclay2015}, enabling us to generalise our proposed algebras and manipulate failure rates directly. Our previous work \citep{Yu2020} sketched relevant ideas, which could be formalised in future research.

Finally we note that the inferential framework we have developed here can be adapted to accommodate natural language data extracted for example from maintenance logs where engineers write about the faults they observe and the possible reasons for what they see. We can embed the causal reasoning behind the texts by mapping them onto the event tree and learning the structure of the causal CEG. This requires natural language processing techniques where the d-events play a role as a bridge to link the tree to the free texts. \cite{Yu2021} and \cite{Yu2022} proposed a naive way to implement this idea which used a hierarchical structure. These can be used to help structure appropriate CEGs like the examples as well as informing the posterior floret probabilities.  More advanced algorithm can be developed in the future to automate this causal learning process.

\bibliographystyle{unsrtnat}


\begin{appendices}
\section*{Supplement to Causal Chain Event Graphs for Remedial Maintenance}

\section{The proof of Theorem 4.5}\label{app:a}
If we can find a partition $\{\Lambda_z\}$ of $\Lambda(\bm{w}^*)$, then the singular intervention causal query conditional on $\Lambda(\bm{w}^*)$ can be written as:
\begin{equation}\label{a1}
    \pi^{\Lambda(\bm{w}^*)}(\Lambda_y||\Lambda_x)=\sum_{z}\pi^{\Lambda(\bm{w}^*)}(\Lambda_y|\hat{\Lambda}_x,\Lambda_z)\pi^{\Lambda(\bm{w}^*)}(\Lambda_z||\Lambda_x).
\end{equation}
If we estimate this intervened quantity from the partially observed system, then
\begin{equation}
    \begin{split}
  \pi^{\Lambda(\bm{w}^*)}(\Lambda_y||\Lambda_x)&= \sum_{\substack{w\in\bm{w^*}\\ e(w,w')\in e(x)}}\pi^{\Lambda(\bm{w}^*)}(\Lambda(w))\pi^{\Lambda(\bm{w}^*)}(\Lambda_y|\Lambda(w')) \\
  &=\sum_{\substack{w\in\bm{w^*}\\ e(w,w')\in e(x)}}\pi^{\Lambda(\bm{w}^*)}(\Lambda(w))\pi^{\Lambda(\bm{w}^*)}(\Lambda_y|\Lambda(e_{w,w'}))\\
  &=\sum_{\substack{w\in\bm{w^*}\\ e(w,w')\in e(x)}}\pi^{\Lambda(\bm{w}^*)}(\Lambda(w))\sum_z\pi^{\Lambda(\bm{w}^*)}(\Lambda_y|\Lambda(e_{w,w'}),\Lambda_z)\pi^{\Lambda(\bm{w}^*)}(\Lambda_z|\Lambda(e_{w,w'}))\\
    \end{split}
\end{equation}
By the two criteria specified in Theorem 4.4 for the back-door partition $\{\Lambda_z\}$, we can replace the last two terms in the last step by the following:
\begin{equation}
\begin{split}
     \pi^{\Lambda(\bm{w}^*)}(\Lambda_y||\Lambda_x) & = \sum_{\substack{w\in\bm{w^*}\\ e(w,w')\in e(x)}}\pi^{\Lambda(\bm{w}^*)}(\Lambda(w))\sum_z\pi^{\Lambda(\bm{w}^*)}(\Lambda_y|\Lambda(w),\Lambda_x,\Lambda_z)\pi^{\Lambda(\bm{w}^*)}(\Lambda_z|\Lambda(w)) \\
     & = \sum_z\pi^{\Lambda(\bm{w}^*)}(\Lambda_y|\Lambda_x,\Lambda_z)\pi^{\Lambda(\bm{w}^*)}(\Lambda_z).
\end{split}
\end{equation}
Comparing with Equation \ref{a1}, we therefore have the following two equivalent expressions.

\begin{equation}
   \pi^{\Lambda(\bm{w}^*)}(\Lambda_y|\hat{\Lambda}_x,\Lambda_z) =\pi^{\Lambda(\bm{w}^*)}(\Lambda_y|\Lambda_x,\Lambda_z),
\end{equation}
\begin{equation}
 \pi^{\Lambda(\bm{w}^*)}(\Lambda_z||\Lambda_x)=   \pi^{\Lambda(\bm{w}^*)}(\Lambda_z).
\end{equation}

\section{Results of the simulation study}\label{app:b}
In Section 5, we used Dirichlet priors for the conditional probabilities $\bm{\theta}_v$. \citep{Collazo2018} suggested treating each Dirichlet hyperparameter $\alpha_{v_i,v_j}$ as the number of phantom units arriving at the child node $v_j$ of $v_i$. Let $\alpha_0$ denote the number of phantom units entering the root node $v_0$ and weigh the edges emanating from the same node equally likely. For example, when $\alpha_0=1$, we have $\alpha_{v_0,v_1}=\alpha_{v_0,v_2} = 0.5$. We perform the analysis mentioned in Section 5 for different values of $\alpha_0$ and check the total situational error ($\varepsilon(\mathcal{T})=\sum_{v\in V_{\mathcal{T}}}||\bm{\theta}^\dagger_v-\tilde{\bm{\theta}}_v||_2$). This error metric is the sum of the Euclidean distance between the true conditional probabilities $\bm{\theta}_v^\dagger$ and the mean posterior probabilities $\tilde{\bm{\theta}}_v$ estimated on the best scoring model for all stages. The results are depicted in Figure \ref{fig:total error}, where we observe that the difference between the total situational errors for different $\alpha_0$ values is quite small. When decomposing the error for each stage, as shown in Figure \ref{fig:situational error}, the curves corresponding to different $\alpha_0$ values overlap. Thus the method is robust to the choice of hyperparameters. Other analyses or applications of the proposed method can be found in \cite{Yu2020} and \cite{Yu2022}.
\begin{figure}
    \centering
    \begin{subfigure}[b]{0.45\textwidth}
\includegraphics[width=\textwidth]{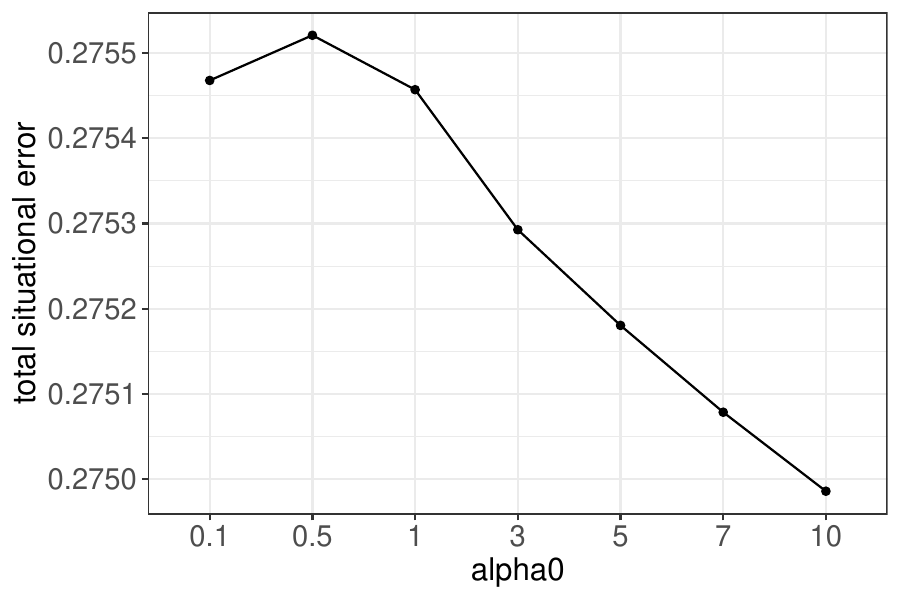}
    \caption{Comparison of total errors for different values of $\alpha_0$.}
    \label{fig:total error}
    \end{subfigure}
    \begin{subfigure}[b]{0.45\textwidth}
\includegraphics[width=\textwidth]{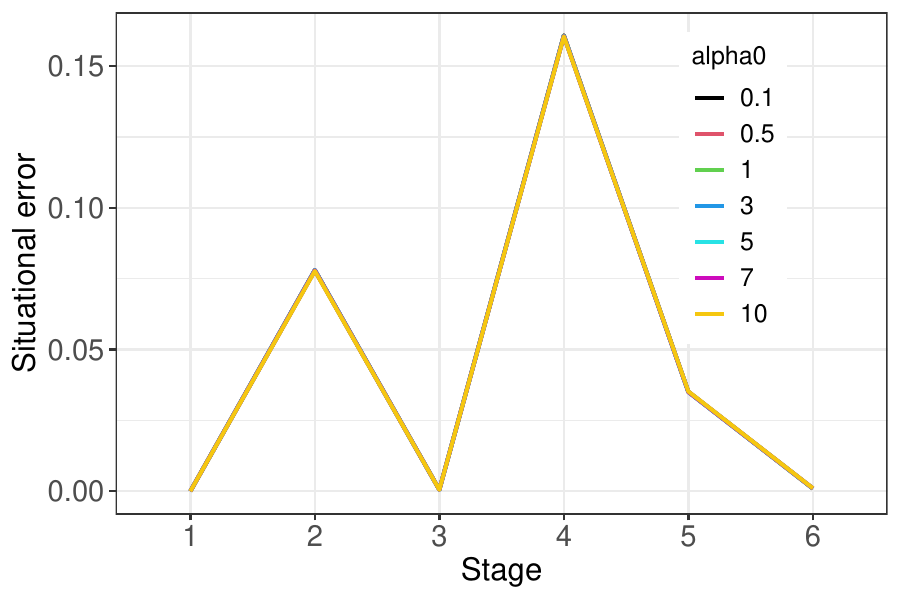}
    \caption{Comparison of situational errors for different values of $\alpha_0$.}
    \label{fig:situational error}
    \end{subfigure}
    \caption{Error plots.}
    \label{fig:error plots}
\end{figure}

\section{Interventions that improve system efficiency}\label{app:c}

Figure \ref{fig:failure life cycle} uses a bathtub curve to portray the life cycle of a unit, reflecting the change in failure rate. The perfect remedial intervention formalised in this paper corresponds to the black curve in the second life cycle. 
If the system is improved after maintenance, then the second life cycle may look like the red curve where the failure rate is reduced throughout the entire lifetime of the unit compared to the first life cycle.

Such an intervention may involve manipulation not only of the probability distribution over root causes: either upstream or downstream florets of the root causes might be manipulated, depending on the maintenance. The algebras proposed in the paper can be generalised to encode such combinatorial manipulations. Moreover, distribution of time-to-event might also be affected. On a dynamic CEG, we can model semi-Markov processes with semi-Markov kernel $Q_{w_i,w_j}(t)=\theta_{w_i,w_j} P(h_{w_i,w_j}\leq t)$  where $h_{w_i,w_j}$ denotes the holding time at $w_i$ just before transitioning to $w_j$. We can specify a parametric distribution for each conditional holding time, for example, a Weibull distribution. Then controlling the shape parameter of the Weibull distribution enables us to work on different phases of the bathtub curve.
\begin{figure}
    \centering
    \includegraphics[width=0.8\textwidth]{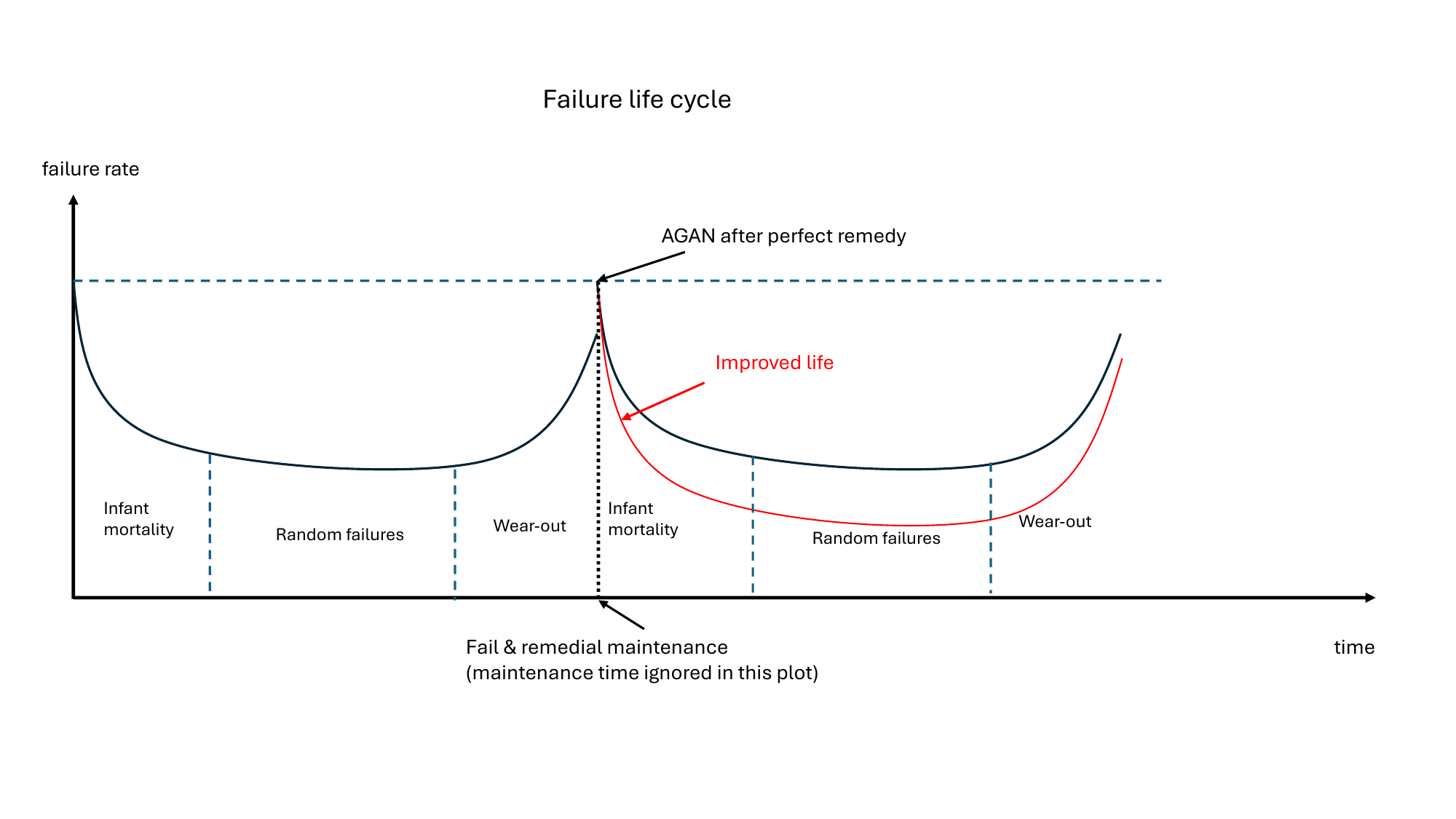}
    \caption{Plot of failure life cycle of a system.}
    \label{fig:failure life cycle}
\end{figure}

\end{appendices}

\end{document}